\documentclass[11pt]{article}

 \usepackage{algorithm}
\usepackage{tcolorbox}
\usepackage{amssymb,xspace,graphicx,relsize,bm,xcolor,amsmath,breqn,algpseudocode,multirow}
\usepackage{tikz,amsmath}
\usepackage{tcolorbox}
\newcommand{\Exp}{\mathbb{E}}

\usepackage[margin=1in]{geometry}
\newcommand{\braketbra}[3]{\langle #1|#2| #3 \rangle}
\newcommand{\E}{\mathcal{E}}
\newcommand{\M}{\mathcal{M}}

\usepackage{graphicx}
\usepackage{amsmath}
\usepackage{amssymb}
\usepackage{amsthm}
\usepackage{dsfont}
\usepackage{array}
\usepackage{makecell}

\newcommand{\Sh}{\ensuremath{\mathcal{S}}}
\newcommand{\X}{\ensuremath{\mathcal{X}}}

\newcommand{\C}{\ensuremath{\mathcal{C}}}

\newcommand{\id}{\ensuremath{\mathbb{I}}}

\usepackage[pagebackref]{hyperref}
\usepackage[usestackEOL]{stackengine}
\usepackage{thm-restate,mathrsfs}
\usepackage{enumerate}
\usepackage{array}
\usepackage{parskip}
\def\01{\{0,1\}}

\newcommand{\ket}[1]{|#1\rangle}
\newcommand{\bra}[1]{\langle#1|}
\newcommand{\ketbra}[2]{|#1\rangle\langle#2|}
\hypersetup{
	colorlinks,
	linkcolor={blue!100!black},
	citecolor={red!100!black},
}

\newcommand{\be}{\begin{equation}}
\newcommand{\ee}{\end{equation}}
\newcommand{\ba}{\begin{array}}
\newcommand{\ea}{\end{array}}
\newcommand{\bea}{\begin{eqnarray}}
\newcommand{\eea}{\end{eqnarray}}

\newcommand{\braket}[3]{\langle{#1}|{#2}|{#3}\rangle}

\usepackage{mathtools}

\newcommand{\ra}{\rangle}
\newcommand{\la}{\langle}
\newcommand{\bias}{\mathrm{bias}}

\newcommand{\Tr}{\text{Tr}}

\newcommand{\norm}[1]{\left\lVert#1\right\rVert}
\newcommand{\mathF}{\mathbb{F}_2}
\newcommand{\Fset}[1]{\{0,1\}^{#1}}

\newcommand{\calD}{{\cal D }}

\newcommand{\calP}{{\cal P }}
\newcommand{\calM}{{\cal M }}
\newcommand{\FF}{\mathbb{F}}
\newcommand{\EE}{\mathbb{E}}
\newcommand{\ZZ}{\mathbb{Z}}

\newcommand{\RR}{\mathbb{R}}

\newtheorem{dfn}{Definition}
\newtheorem{prop}{Proposition}
\newtheorem{lemma}{Lemma}

\newtheorem{corollary}{Corollary}
\newtheorem{fact}{Fact}
\newtheorem{claim}{Claim}

\newtheorem{theorem}{Theorem}

\global\long\def\argmin{\operatornamewithlimits{argmin}}

\title{Optimal algorithms for learning quantum phase states}

\date{}

\author{
Srinivasan Arunachalam$^{1}$\thanks{srinivasan.arunachalam@ibm.com}
\quad 
Sergey Bravyi$^{1}$\thanks{sbravyi@us.ibm.com}
 \quad 
Arkopal Dutt$^{1,2,3}$\thanks{arkopal@mit.edu}
 \quad 
Theodore J. Yoder$^{1}$\thanks{ted.yoder@ibm.com}
 \\ \ \\
\small $^{1}$IBM Quantum, Thomas J Watson Research Center, Yorktown Heights, New York 10598, USA \\
\small $^{2}$MIT-IBM Watson AI Lab, Cambridge, Massachusetts 02142, USA \\
\small $^{3}$Department of Physics, Co-Design Center for Quantum Advantage, \\ 
\small Massachusetts Institute of Technology, Cambridge, Massachusetts 02139, USA
}

\begin{document}
\maketitle
\begin{abstract}
We analyze the complexity of learning $n$-qubit quantum phase states. A degree-$d$ phase state is defined as a superposition of all $2^n$ basis vectors $x$ with amplitudes proportional to $(-1)^{f(x)}$, where $f$ is a degree-$d$ Boolean polynomial over $n$ variables. We show that the sample complexity of learning an unknown degree-$d$ phase state is $\Theta(n^d)$ if we allow separable measurements and $\Theta(n^{d-1})$ if we allow entangled measurements. Our learning algorithm based on separable measurements has runtime $\textsf{poly}(n)$ (for constant $d$) and is well-suited for near-term demonstrations as it  requires only single-qubit measurements in the Pauli $X$ and $Z$ bases. We show similar bounds on the sample complexity for learning generalized phase states with complex-valued  amplitudes. We further consider learning phase states when $f$ has sparsity-$s$, degree-$d$ in its $\FF_2$ representation (with sample complexity $O(2^d sn)$), $f$ has Fourier-degree-$t$ (with sample complexity $O(2^{2t})$), and learning quadratic phase states with $\varepsilon$-global depolarizing noise (with sample complexity $O(n^{1+\varepsilon})$). These learning algorithms give us a  procedure to learn the diagonal unitaries of the Clifford hierarchy and IQP~circuits. 
\end{abstract}

\newpage

\section{Introduction} \label{sec:intro}
Quantum state tomography is the problem of learning an unknown quantum state $\rho$ drawn from a specified class of states by performing measurements on multiple copies of $\rho$. The preeminence of this problem in verification of quantum experiments has motivated an in-depth study of state tomography protocols and their limitations for various classes of quantum states~\cite{haah2017sample,o2016efficient,apeldoorn,yuen2022improved}. The main figure of merit characterizing a state tomography protocol is its {\em sample complexity} defined as the number of copies of $\rho$ consumed by the protocol in order to learn $\rho$. Of particular interest are classes of $n$-qubit quantum states that can be learned efficiently, such that the sample complexity grows only polynomially with $n$. Known examples of  efficiently learnable states include Matrix Product States describing weakly entangled quantum spin chains~\cite{cramer2010efficient}, output states of Clifford circuits~\cite{montanaro2017learning}, output states of Clifford circuits with a single layer of $T$~gates~\cite{lai2022learning}, and  high-temperature Gibbs states of local Hamiltonians~\cite{anshu2021sample,haah2021optimal}. Apart from their potential use in  experiments, efficiently learnable quantum states are of great importance for quantum algorithm design. For example, a quantum algorithm for solving the dihedral hidden subgroup problem~\cite{bacon2005optimal} can be viewed as a tomography protocol for learning so-called hidden subgroup states (although this protocol is efficient in term of its sample complexity, its runtime is believed to be super-polynomial~\cite{bacon2005optimal}).

A natural question to then ask is: What are other classes of $n$-qubit quantum states that are ubiquitous in quantum computing, which can be learned efficiently? In this work, we consider the problem of state tomography for {\em phase states} associated with (generalized) Boolean functions. Phase states are encountered in quantum information theory~\cite{hein2004multiparty}, quantum algorithm design~\cite{bacon2005optimal}, quantum cryptography~\cite{ji2018pseudorandom,brakerski2019pseudo}, and quantum-advantage experiments~\cite{bremner2011classical,bremner2017achieving}.

By definition, an $n$-qubit, degree-$d$ phase state has the form
\be
\label{binary_phase_state}
|\psi_f\ra = 2^{-n/2} \sum_{x\in \{0,1\}^n}\;  (-1)^{f(x)} |x\ra,
\ee
where $f\, : \, \{0,1\}^n \to \{0,1\}$ is a  degree-$d$ polynomial, that is,
\be
\label{binary_polynomial}
f(x) = \sum_{\substack{J\subseteq [n],\, |J|\le d}} \alpha_J \prod_{j\in J} x_j {\pmod 2},
\ee
for some coefficients $\alpha_J\in \{0,1\}$. Phase states associated with homogeneous degree-$2$ polynomials $f(x)$ coincide with graph states that play a prominent role in quantum information theory~\cite{hein2004multiparty}. Such states can be alternatively represented as 
\[
|\psi_f\ra = \prod_{(i,j)\in E} \mathsf{CZ}_{i,j} |+\ra^{\otimes n},
\]
where $n$ qubits live at vertices of a graph, $E$ is the set of graph edges,
$\mathsf{CZ}_{i,j}$ is the controlled-$Z$ gate applied to qubits $i,j$, and  $|+\ra=(|0\ra+|1\ra)/\sqrt{2}$. It is known that the output state of any Clifford circuit is locally equivalent to a graph state for a suitable graph~\cite{schlingemann2001stabilizer}. Our results imply that graph states can be learned efficiently using only single-qubit gates and measurements. The best previously known protocol for learning graph states~\cite{montanaro2017learning} requires entangled measurements across two copies of $\ket{\psi_f}$. Other examples of circuits producing phase states include measurement-based quantum computing \cite{rossi2013quantum} and a subclass of $\mathsf{IQP}$ circuits (Instantaneous Quantum Polynomial-time), which correspond to degree-$3$ phase states \cite{montanaro2017circuits}. $\mathsf{IQP}$ circuits are prevalent in quantum-advantage experiments \cite{bremner2011classical,bremner2017achieving} and are believed to be hard to simulate classically.

We also consider generalized degree-$d$ phase states 
\be
\label{generalized_phase_state}
\ket{\psi_f} =
2^{-n/2} \sum_{x\in \{0,1\}^n}\;  
\omega_q^{f(x)} \ket{x},  \qquad \omega_q=e^{2\pi i /q}
\ee
where $q\ge 2$ is an even integer and 
$f\, : \, \{0,1\}^n\to \ZZ_q$ is a degree-$d$ polynomial, that is,
\be
\label{generalized_polynomial}
f(x) = \sum_{\substack{J\subseteq [n],\, |J|\le d}} \alpha_J \prod_{j\in J} x_j {\pmod {q}}.
\ee
for coefficients $\alpha_J\in \ZZ_q = \{0,1,\ldots,q-1\}$. 
It is also known that generalized degree-$d$ phase states with $q=2^d$ can be prepared from diagonal unitary operators~\cite{cui2017diagonal} in the $d$-th level of the Clifford hierarchy~\cite{gottesman1999demonstrating}. Additionally, it is known that the output state of a random $n$-qubit Clifford circuit is a generalized $q=4$, degree-$2$ phase state with a constant probability~\cite[Appendix~D]{bravyi2016improved}.
 Binary and generalized phase states have also found applications in cryptography~\cite{ji2018pseudorandom,brakerski2019pseudo}, and complexity theory~\cite{irani2021quantum} (we discuss this in the next section).

In this work, we consider learning phase states through two types of tomography protocols based on {\em separable} and {\em entangled} measurements. The former can be realized as a sequence of $M$ independent measurements, each performed on a separate copy of $|\psi_f\ra$ (furthermore our learning algorithms only require single \emph{qubit} measurements).  The latter performs a joint measurement on the state $|\psi_f\ra^{\otimes M}$. Our goal is to then derive upper and lower bounds on the sample complexity~$M$ of learning $f$, as a function of $n$ and $d$. In the next section, we state our main results. Interestingly, our protocols based on separable measurements require only single-qubit gates and single-qubit measurements making them well suited for near-term demonstrations.

\subsection{Summary of contributions and applications}
We first introduce some notation before giving an overview of our contributions. For every $n$ and $d\leq n/2$, let $\calP(n,d)$ be the set of all degree-$d$ polynomials of the form Eq.~(\ref{binary_polynomial}). Let $\calP_q(n,d)$ be the set of all degree-$d$ $\ZZ_q$-valued polynomials of the form Eq.~(\ref{generalized_phase_state}). By definition, $\calP_2(n,d)\equiv \calP(n,d)$. To avoid confusion, we shall refer to states defined in Eq.~(\ref{binary_phase_state}) as binary phase states and in Eq.~(\ref{generalized_phase_state}) as generalized phase states. Our learning protocol takes as input integers $n,d$ and $M$ copies of a degree-$d$ phase state $\ket{\psi_f}$ with unknown $f\in \calP(n,d)$ (or $f\in \calP_q(n,d)$). The protocol outputs a classical description of a polynomial $g\in \calP(n,d)$ (or $g\in \calP_q(n,d)$) such that $f=g$ with high~probability. 

The main result in this work are optimal algorithms for learning phase states if the algorithm is allowed to make separable or entangled measurements. Prior to our work, we are aware of only two works in this direction (i)  algorithms for efficiently learning degree-$1$ and degree-$2$ phase states; (ii) Montanaro~\cite{montanaro2012quantum} considered learning multilinear polynomials $f$, assuming we have \emph{query access} to $f$, which is a stronger learning model than the sample access model that we assume for our learning algorithm. In this work, we show that if allowed separable measurements, the \emph{sample} complexity of learning binary phase states and generalized phase states is $O(n^d)$. If allowed entangled measurements, we obtain a sample complexity of $O(dn^{d-1})$ for learning binary phase states. We further consider settings where the unknown function $f$ we are trying to learn is known to be sparse, has a small Fourier-degree and the setting when given noisy copies of the quantum phase state.  In Table~\ref{tab:summary_results_paper}, we summarize all our main results (except the first two rows, which include the main prior work in this direction).

\begin{table}[h]
\small
\centering
 \begin{tabular}{|c | c  | c  |c|} 
 \hline
  & \makecell{Sample complexity} & \makecell{Time complexity}  & Measurements\\ [0.5ex] 
 \hline
 \makecell{ Binary phase state $\FF_2$-degree-$1$~\cite{bernstein1997quantum}} & $\Theta(1)$  & $O(n^3)$   &Separable\\ 
 \hline
 \makecell{Binary phase state $\FF_2$-degree-$2$\\~\cite{montanaro2017learning,rotteler2009quantum}}   & $O(n)$ & $O(n^3)$ &Entangled \\ \hline
 \makecell{ Binary phase state $\FF_2$-degree-$d$} & \makecell{$\Theta(n^{d})$\\ Theorem~\ref{thm:binaryphasedegreed},~\ref{thm:lowerboundbinaryphase}} & $O(n^{3d-2})$ &Separable \\ 
 \hline
  \makecell{Binary phase state $\FF_2$-degree-$d$ } & \makecell{$\Theta(n^{d-1})$\\Theorem~\ref{thm:entangledupperbound}}  & $O(\exp(n^d \log 2))$ &Entangled \\ \hline
  \makecell{Generalized phase states degree-$d$} & \makecell{$\Theta(n^{d})$\\Theorem~\ref{thm:generalizedphasedegreed}} & $O(\exp(n^d \log q))$ &Separable \\
   \hline
    \makecell{\emph{Sparse} Binary phase state\\ $\FF_2$-degree-$d$, $\FF_2$-sparsity $s$} & \makecell{${O}(2^d s n )$ \\Theorem~\ref{thm:sample_complexity_sparse_degree_d_polynomial_qs}} & $O(2^{3d}s^3 n)$ &Separable \\ \hline
    \makecell{Binary phase state $\FF_2$-degree-$2$\\ with global depolarizing noise $\varepsilon$} & \makecell{$n^{1+O(\varepsilon)}$ \\Theorem~\ref{lem:noisystablearning}} & $O(2^{n/\log n})$ &Entangled\\ \hline
    \makecell{Binary phase state $\FF_2$-degree-$2$\\ with local depolarizing noise $\varepsilon$} & \makecell{$\Theta((1-\varepsilon)^n)$ \\Theorem~\ref{thm:localdepolarizing}}& \makecell{$O(2^{n/\log n})$} & Entangled\\ \hline
    \makecell{ Binary phase state Fourier-degree-$d$} & \makecell{$O(2^{2d})$ \\ Theorem~\ref{thm:fourierbinaryphasedegreed}} & \makecell{$O(\exp(n^2))$} &Entangled \\ 
 \hline
 \end{tabular}
 \caption{Upper and lower bounds of sample complexity for exact learning of $n$-qubit phase states with degree-$d$. For precise statements of the bounds, we refer the reader to the theorem statements.}
 \label{tab:summary_results_paper}
\end{table}
\newpage 
Before we give a proof sketch of these results, we first discuss a couple of motivations for considering the task of learning phase states and corresponding applications.  

\paragraph{Quantum complexity.} Recently, there has been a few results in quantum cryptography~\cite{ji2018pseudorandom,ananth2021cryptography,brakerski2019pseudo} and complexity theory~\cite{irani2021quantum} which used the notion of phase states. 

Ji et al.~\cite{ji2018pseudorandom} introduced the notion of \emph{pseudorandom quantum states} as states of the form $\ket{\phi}=\frac{1}{\sqrt{2^n}}\sum_{x\in \01^n}\omega_N^{F(x)}\ket{x}$ where $F$ is a pseudorandom function.\footnote{We do not discuss the details of pseudorandom functions here, we refer the interested reader to~\cite{ji2018pseudorandom}.} Ji et al.~showed that states of the form $\ket{\phi}$ are efficiently preparable and statistically indistinguishable from a Haar random state, which given as input to a polynomial-time quantum algorithm. A subsequent work of Brakerski~\cite{brakerski2019pseudo} showed that it suffices to consider $\ket{\phi'}=\frac{1}{\sqrt{2^n}}\sum_{x\in \01^n}(-1)^{F(x)}\ket{x}$ (where $F$ again is a pseudorandom function) and  such states are also efficiently preparable and statistically indistinguishable from Haar random states. Subsequently, these states have found applications in proposing many  cryptosystems~\cite{ananth2021cryptography}. Although none of these works discuss the degree of the phase function $F$, our result shows implicitly that when $F$ is low-degree, then $\ket{\phi}$ is exactly learnable and hence distinguishable from Haar random states, implying that they cannot be quantum pseudorandom states. 
In another recent work, Irani et al.~\cite{irani2021quantum} considered the power of quantum witnesses in proof systems. In particular, they showed that in order to construct the witness to a \textsf{QMA} complete problem, say the ground state $\ket{\phi}$ to a local-Hamiltonian problem, it suffices to consider a phase state $\frac{1}{\sqrt{2^n}}\sum_x (-1)^{F(x)}\ket{x}$ which has a good overlap to $\ket{\phi}$. To this end, they show a strong property that, for \emph{every} state $\ket{\tau}$ and a random Clifford operator $U$ (or, more generally, an element of some unitary $2$-design), the state $U\ket{\tau}$ has constant overlap with a phase state~\cite[Lemma~A.5]{irani2021quantum}.
Our learning result implicitly shows that, assuming  $\textsf{QMA}\neq \textsf{QCMA}$, then the phase state that has constant overlap with the ground space energy of the local Hamiltonian problem, cannot be of low degree.

\paragraph{Learning quantum circuits.} Given access to a quantum circuit $U$, the goal of this learning task is to learn a circuit representation of $U$. The sample complexity for learning a general $n$-qubit quantum circuit is known to be $2^{\Theta(n)}$ \cite{chuang1997prescription,mohseni2008quantum}, which is usually impractical.

If we restrict ourselves to particular classes of quantum circuits, there are some known results for efficient learnability. Low \cite{low2009learning} showed that an $n$-qubit Clifford circuit can be learned using $O(n)$ samples. However, this result was only an existential proof and requires access to the conjugate of the circuit. Constructive algorithms were given in Low \cite{low2009learning}, and Lai and Cheng \cite{lai2022learning}, both of which showed that Clifford circuits can be learned using $O(n^2)$ samples. Both these algorithms require entangled measurements with the former algorithm using pretty-good measurement~\cite{harrow2012many}, and the latter using Bell sampling. In this work, we show that Clifford circuits producing degree-$2$ binary phase states, can be learned in $O(n^2)$ samples, matching their result but only using separable measurements. Moreover, Low \cite{low2009learning} also gave an existential proof of algorithms for learning circuits in the $d$-th level of the Clifford hierarchy, using $O(n^{d-1})$ samples. In this work, we give constructive algorithms for learning the diagonal elements of the Clifford hierarchy in $O(n^d)$ samples using separable measurements. A direct result of this is that a subset of $\textsf{IQP}$ circuits, which are also believed to be hard to simulate classically \cite{bremner2011classical,bremner2016average}, are shown to be efficiently learnable. Our learning result thus gives an efficient method for verifying $\textsf{IQP}$ circuits that may be part of quantum-advantage experiments \cite{bremner2017achieving,novo2021quantum}.

\paragraph{Learning hypergraph states.} We finally observe that degree-$3$ (and higher-degree) phase states have appeared in works~\cite{rossi2013quantum,takeuchi2019quantum} on measurement-based quantum computing (MBQC), wherein they refer to these states as \emph{hypergraph states}. These works show that single-qubit measurements in the Pauli $X$ or $Z$ basis performed on a suitable degree-3 hypergraph state are sufficient for universal MBQC. Our learning algorithm gives a procedure for learning these states in polynomial-time and could potentially be used as a subroutine for verifying MBQC.

\subsection{Proof sketch}
In this section we briefly sketch the proofs of our main results. 
\subsubsection{Binary phase states} 

As we mentioned earlier, Montanaro~\cite{montanaro2017learning} and Roettler~\cite{rotteler2009quantum} showed how to learn degree-$2$ phase states using $O(n)$ copies of the state. Crucial to both their learning algorithms was the following so-called Bell-sampling procedure: given two copies of $\ket{\psi_f}=\frac{1}{\sqrt{2^n}}\sum_{x} (-1)^{f(x)}\ket{x}$ where $f(x)=x^\top A x$ (where $A\in \FF_2^{n\times n}$), perform $n$ CNOTs from the first copy to the second, and measure the second copy. One obtains a uniformly random $y\in \FF_2^n$ and the state 
$$
\frac{1}{\sqrt{2^n}}\sum_x (-1)^{f(x)+f(x+y)}\ket{x}=\frac{(-1)^{y^\top Ay}}{\sqrt{2^n}}\sum_x (-1)^{x^\top(A+A^\top)\cdot y}\ket{x}.
$$
Using Bernstein-Vazirani~\cite{bernstein1997quantum} one can apply $n$-qubit Hadamard transform to obtain the bit string $(A+A^\top)\cdot y$. Repeating this process $O(n\log n)$ many times, one can learn $n$ linearly independent constraints about $A$, and along with Gaussian elimination, allows one to learn~$A$.\footnote{It remains to learn the diagonal elements of $A$, but one can learn those using an extra step, which we discuss further in Theorem~\ref{lem:noisystablearning}.}  

Applying this same Bell-sampling procedure to degree-$3$ phase states does not easily learn the phase function. In this direction, from two copies of the degree-3 phase state $\ket{\psi_f}$ one obtains a uniformly random $y\in\FF_2^n$ and the state $\ket{\psi_{g_y}}=\frac{1}{\sqrt{2^n}}\sum_x (-1)^{g_y(x)}\ket{x}$ for a degree-$2$ polynomial $g_y(x)=f(x)+f(x+y)$. One might now hope to apply the degree-$2$ learning algorithm from above, but since the single copy of $\ket{\psi_{g_y}}$ was randomly generated, it takes $\Omega(\sqrt{2^n})$ copies of $\ket{\psi_f}$ to obtain enough copies of $\ket{\psi_{g_y}}$. Our main idea is to circumvent this Bell-sampling approach and instead propose two techniques that allow us to learn binary phase states using separable and entangled measurements which we discuss further below.  

\paragraph{Separable measurements, upper bound.}
Our first result is that we are able to learn binary phase states using separable measurements with sample complexity $O(n^d)$. In order to prove our upper bounds of sample complexity for learning with separable measurements, we make a simple observation. Given one copy of $\ket{\psi_f}=\frac{1}{\sqrt{2^n}}\sum_x (-1)^{f(x)}\ket{x}$, measure qubits $2,3,\ldots,n$ in the computational basis. Suppose the resulting string is $y \in \{0,1\}^{n-1}$. The post-measurement state of qubit $1$ is then given~by  
$$
\ket{\psi_{f,y}} = \frac{1}{\sqrt{2}} \left[ (-1)^{f(0y)}\ket{0} + (-1)^{f(1y)} \ket{1}\right].
$$
By applying a Hadamard transform to  $\ket{\psi_{f,y}}$ and measuring, the algorithm obtains $p_1(y)=f(0y)+f(1y) \mod 2$, which can be viewed as the derivative of $f$ in the first direction at point~$y$. Furthermore observe that $p_1$ is a degree $\leq d-1$ polynomial over $(n-1)$ variables. Hence, the learning algorithm repeatedly measures the last $(n-1)$ qubits and obtains $y^{(1)},\ldots,y^{(M)}$ for $M=n^{d-1}$ and obtains $(y^{(k)},p_1(y^{(k)}))$ for all $k=1,2,\ldots,M$ using the procedure above, which suffices to learn~$p_1$ completely. Then the algorithm repeats the same procedure by measuring all the qubits except the second qubit in the computational basis and learns the derivative of $f$ in the second direction. This is repeated over all the $n$ qubits. Through this procedure, a learning algorithm learns the partial derivatives of $f$ in the $n$ directions and a simple argument shows that this is sufficient to learn $f$ completely. This gives an overall sample complexity of $O(n^d)$. The procedure above only uses single qubit measurements in the $\{X,Z\}$ basis.  

\paragraph{Separable measurements, lower bound} Given the algorithm for learning binary phase states using separable measurements, a natural question is: Is the upper bound on sample complexity we presented above tight? Furthermore, suppose the learning algorithm was allowed to make arbitrary $n$-qubit measurements on a {single copy} of $\ket{\psi_f}$, instead of \emph{single qubit} measurements (which are weaker than single \emph{copy} measurements), then could we potentially learn $f$ using fewer than $O(n^d)$~copies?

Here we show that if we allowed \emph{arbitrary} single copy measurements, then a learning algorithm needs $\Omega(n^d)$ many copies of $\ket{\psi_f}$  to learn $f$. In order to prove this lower bound, our main technical idea is the following. Let $f$ be 
a degree-$d$ polynomial with $n$ variables sampled uniformly at random. 
Suppose a learning algorithm measures the phase state $\ket{\psi_f}$ in an arbitrary orthonormal basis $\{U\ket{x}\}_x$.
We show that 
the distribution describing the measurement outcome $x$ is ``fairly" uniform. In particular, 
\begin{align}
\label{eq:abstractentropy}
    \mathop{\mathbb{E}}_f [H(x|f)]\ge n - O(1),
\end{align} 
where $H(x|f)$ is the Shannon entropy of
a distribution $P(x|f)=|\la x|U^*|\psi_f\ra|^2$.
Thus, for a typical $f$, measuring one copy
of the phase state $\ket{\psi_f}$ provides at most $O(1)$ bits of information about $f$. Since a random uniform
degree-$d$ polynomial $f$ with $n$ variables has entropy~$\Omega(n^d)$, one has to measure $\Omega(n^d)$
copies of $\psi_f$ in order to learn $f$.
To prove Eq.~\eqref{eq:abstractentropy}, we first lower bound the Shannon entropy by Renyi-two entropy
and bound the latter by deriving 
an explicit formula for $\EE_f [|\psi_f\ra\la \psi_f|^{\otimes 2}]$.

\paragraph{Entangled measurements.} After settling the sample complexity of learning binary phase states using separable measurements, one final question question remains: Do entangled measurements help in reducing the sample complexity? For the case of quadratic polynomials, we know that Bell measurements (which are entangled measurements) can be used to learn these states in sample complexity $O(n)$. However, as mentioned earlier, it is unclear how to extend the Bell measurement procedure for learning larger degree polynomials.

Here, we give a learning algorithm based on the so-called pretty-good measurements (PGM) that learns $\ket{\psi_f}$ for a degree-$d$ polynomial $f$
using $O(n^{d-1})$ copies of $\ket{\psi_f}$. In order to prove this bound, we follow the following three step approach: (a) we first observe that in order to learn degree-$d$ binary phase states, the \emph{optimal} measurement is the pretty good measurement since the ensemble $\Sh=\{\ket{\psi_f}\}_f$ is geometrically uniform. By geometrically uniform, we mean that $\Sh$ can be written as $\Sh=\{U_f\ket{\phi}\}_{f}$ where $\{U_f\}_f$ is an Abelian group. 
(b) We next observe a property about the geometrically uniform state identification problem (which is new as far as we are aware): suppose $\Sh$ is a geometrically uniform ensemble, then the success probability of the PGM in correctly identifying $f$, given copies of $\ket{\psi_f}$, is \emph{independent} of $f$, i.e., every element of the ensemble has the same probability of being identified correctly when measured using the PGM. (c) Finally, we need one powerful tool regarding the the weight distribution of Boolean polynomials: it was shown in~\cite{abbe2015reed} that for any degree-$d$ polynomial $f$, the following relation on $\textsf{wt}(f)$ or the fraction of strings in $\{0,1\}^n$ for which $f$ is non-zero holds:
$$|\{f\in \calP(n,d):\textsf{wt}(f)\leq (1-\varepsilon)2^{-\ell}\}|\leq (1/\varepsilon)^{C\ell^4\cdot  \binom{n-\ell}{\leq d-\ell}},
$$
for every $\varepsilon\in (0,1/2)$ and $\ell\in \{1,\ldots,d-1\}$. Using this statement, we can comment on the average inner product of $|\langle \psi_f|\psi_g\rangle|$ over all ensemble members with $f \neq g \in \calP(n,d)$. Combining this with a well-known result of PGMs, we are able to show that, given $M=O(n^{d-1})$ copies of $\ket{\psi_f}$ for $f\in \Sh$, the PGM identifies $f$ with probability $\geq 0.99$. Combining observations (a) and (b), the PGM also has the same probability of acceptance given an arbitrary $f\in \Sh$. Hence, we get an overall upper bound of $O(n^{d-1})$ for sample complexity of learning binary phase states using entangled measurements. 

The lower bound for entangled measurement setting is straightforward: each quantum sample $\frac{1}{\sqrt{2^n}}\sum_{x\in \{0,1\}^n}(-1)^{f(x)}\ket{x}$ provides $n$ bits of information and the goal is to learn $f$ which contains $O(n^d)$ bits of information, hence by Holevo's bound, we need at least $n^{d-1}$ quantum samples in order to learn $f$ with high probability.

\paragraph{Implications for property testing.} We remark that our learning algorithm can also be used in a naive way for property testing phase states. Let $\C=\{\ket{\psi_f}:f\in \calP(n,d)\}$ be the class of degree-$d$ phase states. The property testing question is: How many copies of an unknown $\ket{\phi}$ is sufficient to decide if $\ket{\phi}\in \C$ or $\min_{\ket{\psi_f}\in \C}\|\ket{\phi}-\ket{\psi_f}\|_2\geq 1/3$? As far as we are aware, the only prior work in this direction is when $d=1$ (using Bernstein-Vazirani~\cite{bernstein1997quantum}) and $d=2$ (using~\cite{gross2021schur} which shows how to solve this task using $6$ copies), but for larger $d$ it is unclear what is the sample complexity. It is also unclear how to perform the property testing task (even for $d=2$) using just \emph{separable measurements}. Using our learning result, we get the following: take $n^d$ copies of $\ket{\phi}$ and run our learning procedure using separable measurements.\footnote{We could also use $n^{d-1}$ copies of $\ket{\phi}$ and run our learning procedure using entangled measurements.} If $\ket{\phi}=\ket{\psi_f}$, then our algorithm learns~$f$. If $\ket{\phi}$ is not a phase state, the algorithm may fail, in which case the test classifies the state as a non-phase state. The worst case is if the algorithm succeeds and learns some incorrect phase state $\ket{\psi_f}$ from the non-phase input state. So, after running the learning algorithm and obtaining $\ket{\psi_f}$, use $O(1)$ more copies of $\ket{\phi}$ and run a swap test between $\ket{\phi}$ and $\ket{\psi_f}$, which succeeds with probability~$1$ if $\ket{\phi}=\ket{\psi_f}$ and rejects with probability at least $\Omega(1)$ if $\min_{\ket{\psi_f}\in \C}\|\ket{\phi}-\ket{\psi_f}\|_2\geq 1/3$. 

\subsubsection{Generalized phase states}
As far as we are aware, ours is the first work that considers the learnability of generalized phase states (using either entangled or separable measurements). The sample complexity upper bounds follow the same high-level idea as that in the binary phase state setting. However, we need a few more technical tools for the generalized setting which we discuss below. 

\paragraph{Separable bounds.} At a high-level, the learning procedure for generalized phase states is similar to the procedure for learning binary phase states with the exception of a couple of subtleties that we need to handle here. Suppose we perform the same procedure as in binary phase states by measuring the last $(n-1)$ qubits in the computational basis. 
We then obtain a uniformly random $y\in \FF_2^{n-1}$, and the post-measurement state for a generalized phase state is given by
$$
    \ket{\psi_{f,y}} =\frac1{\sqrt{2}} ( \omega_q^{f(0y)} \ket{0} + \omega_q^{f(1y)} \ket{1}).
$$
This state is proportional to $(|0\ra+\omega_q^c|1\ra)/\sqrt{2}$, where $c=f(1y)-f(0y) {\pmod q}$.
In the binary case, $q=2$, the states associated with $c=0$ and $c=1$ are orthogonal, so that the value of $c$
can be learned with certainty by measuring $\ket{\psi_{f,y}}$  in the Pauli $X$ basis.
However, in the generalized case, $q>2$, the states $(|0\ra+\omega_q^c|1\ra)/\sqrt{2}$ with $c\in \ZZ_q$
are not pairwise orthogonal. It is then unclear how to learn $c$ given a single copy of  $\ket{\psi_{f,y}}$.
However, we observe that it is still possible to obtain a value $b\in \ZZ_q$ such that $b\neq c$ with certainty. 
To this end, consider a POVM whose elements are given by $\M=\{\ketbra{\phi_b}{\phi_b}\}_{b\in \ZZ_q}$,
where $\ket{\phi_b}=\frac1{\sqrt{2}} ( \ket{0} - \omega_q^{b} \ket{1})$.
 Applying this POVM $\M$ onto an unknown state $(|0\ra+\omega_q^c|1\ra)/\sqrt{2}$ we observe that $c$ is the outcome with probability $0$ and furthermore \emph{every} other outcome $b\ne c$ appears with non-negligible probability $\Omega(q^{-3})$.

Hence with one copy of $\frac{1}{\sqrt{2^n}}\sum_{x\in \Fset{n}}\omega_q^{f(x)}\ket{x}$, we obtain uniformly random $y\in \Fset{n-1}$ and $b\in \ZZ_q$ such that $f(1y)-f(0y)\neq b$. We now repeat this process $m=O(n^{d-1})$ many times and obtain $(y^{(k)},b^{(k)})$ for $k=1,2,\ldots,M$ such that $f(1y^{(k)})-f(0y^{(k)})\neq b^{(k)}$ for all $k \in [M]$. We next show a variant of the Schwartz-Zippel lemma in the following sense: that for every $f\in \calP_q(n,d)$ and $c\in \ZZ_q$, then either $f$ is a constant function or the fraction of $x\in \FF_2^n$ for which $f(x)\neq c$ is at least $2^{-d}$. Using this, we show that after obtaining $O(2^d n^{d-1})$ samples, we can find a polynomial $g\in \calP_q(n-1,d-1)$ for which $f(1y)-f(0y)=g(y)$. We now repeat this protocol for $n$ different directions (by measuring each of the $n$ qubits in every iteration) and we learn all the $n$ directional derivatives of $f$, which suffices to learn $f$ completely.

\paragraph{Entangled bounds.} We do not give a result on learning generalized phase states with entangled measurements. We expect the proof of the sample complexity upper bound for learning generalized phase states using entangled measurements should proceed similarly to our earlier analysis of learning binary phase states using entangled measurements. However, we need a new technical tool that generalizes the earlier work on the weight distribution \cite{abbe2020reed} of Boolean functions $f: \FF_2^n \rightarrow \FF_2$ to those of form $f: \FF_2^n \rightarrow \ZZ_q$ with $q=2^d$.

\subsubsection{Learning with further constraints}

\paragraph{Learning sparse and low-Fourier degree states.} A natural constraint to put on top of having low $\FF_2$-degree in the polynomial is the sparsity, i.e., number of monomials in the $\FF_2$ decomposition of $f$. Sparse low-degree phase states appear naturally when learning circuits with few gates. In particular, suppose we are learning a quantum circuit $U$ with $s$ gates from $\{\mathsf{Z},\mathsf{CZ},\ldots,\mathsf{C}^{d-1}\mathsf{Z}\}$ (where $\mathsf{C}^{m}\mathsf{Z}$ is the controlled-$Z$ gate with $m$ controls), then the output of $U\ket{+}^{\otimes n}$ is a phase state with sparsity-$s$ and degree-$d$.

One naive approach to learn sparse $\FF_2$ polynomials is to directly apply our earlier learning algorithm for binary phase states but this ignores the $\FF_2$-sparsity information, and doesn't improve the sample complexity. Instead, here we use ideas from compressed sensing~\cite{draper2009compressed} to propose a linear program that allows us to improve the sample complexity to $O(2^{d}sn)$. Finally we make an observation that, if the function has \emph{Fourier}-degree $d$, then one can learn $f$, given only $O(2^d\log n)$ many copies of $\ket{\psi_f}$, basically using the fact that there are only $2^{2^d}$ many such functions, each having at least a $2^{-d}$ distance between them. 

\paragraph{Learning with depolarizing noise.} One motivation for learning stabilizer states was potential experimental demonstrations of the learning algorithm~\cite{rocchetto2019experimental}. Here, we consider a theoretical framework in order to understand the sample complexity of learning degree-$2$ phase states under global and local depolarizing noise. In this direction, we present two results. Under global depolarizing noise, i.e., when we are given $\rho_f=(1-\varepsilon)\ketbra{\psi_f}{\psi_f}+\varepsilon\cdot \id$, then it suffices to take $O(n^{1+\varepsilon})$ many copies $\rho_f$ in order to learn $f$. The crucial observation is that one can use Bell sampling to reduce learning $\rho_f$ to learning parities with noise, which we can accomplish using $O(n^{1+\varepsilon})$ samples and in time $2^{n/(\log \log n)}$~\cite{lyubashevsky2005parity}. Additionally, however, a simple argument reveals that under local depolarizing noise, the sample complexity of learning stabilizer states is exponential in~$n$.  

\subsection{Open questions}
Our work leaves open a few interesting questions.

\emph{Improving runtime.}
While our algorithms for learning phase states are optimal in terms of the sample complexity, their runtime scales polynomially with the number of qubits only in the case of binary phase states and separable measurements. 
It remains to be seen whether a polynomial runtime can be achieved in the remaining cases, i.e., learning binary phase states with entangled measurements and generalized phase states with either separable or entangled~measurements. 

\emph{Quantum advantage.}
Suppose $U$ is a polynomial size quantum circuit such that 
$U|0^n\ra$ is a low-degree phase state associated with some
Boolean function $f\, :\, \{0,1\}^n\to \{0,1\}$.
Our results imply there exists an efficient quantum algorithm
that learns $f$ given a classical description of $U$.
An interesting open question
is whether the problem of learning $f$ given a description of $U$
is classically hard. If this is the case, our results would imply
a quantum advantage for the considered learning~task.

\emph{Property testing.} What is the sample complexity of property testing phase states? Given $M$ copies of $\ket{\phi}$ with the promise that either $\ket{\phi}$ is a degree-$d$ phase state or $\varepsilon$-far from the set of degree-$d$ phase states, what is an upper and lower bound on $M$? For $d=1$, we can learn the entire state using $M=1$ copy and for $d=2$, Gross et al.~\cite{gross2021schur} showed that $M=6$ copies suffice for this testing question. For larger $d$, understanding the complexity of testing phase states is an intriguing open question left open by our work, in particular does the sample complexity of testing $n$-qubit degree-$d$ phase states scale as $n^{d-2}$ (for $d\geq 2$) or does it scale as $\textsf{poly}(c^d,n)$ for some $c>1$? 

\emph{Learning more expressive quantum states.} 
We leave as an open question whether our learning algorithms can be extended to binary phase states with a small {\em algebraic degree}. Such states have amplitudes proportional to $(-1)^{\mathsf{tr}F(x)}$, where
$F(x)=\sum_{i=0}^d a_i x^i$ is a degree-$d$ polynomial with coefficients  $a_i\in \FF_{2^n}$ and 
$\mathsf{tr}\, : \, \FF_{2^n} \to \FF_2$ is the trace function defined 
as $\mathsf{tr}(x) = \sum_{j=0}^{n-1} x^{2^j}$.
Here all arithmetic operations use the field $\FF_{2^n}$.
What is the sample complexity of learning $n$-qubit states produced by circuits containing \emph{non-diagonal} unitaries in the $k$-th level in the Clifford hierarchy, on the $\ket{+}^n$ input?  Similarly, what is the complexity of learning a state which has stabilizer rank $k$?\footnote{We know how to learn stabilizer states and stabilizer-rank $2$ states in polynomial time, what is the complexity as a function of rank-$k$?} Similarly can we PAC learn these classes of quantum states in polynomial time?\footnote{For stabilizer circuits, we have both positive and negative results in this direction~\cite{rocchetto2017stabiliser,liang2022clifford} but for more generalized circuits, it remains an open question.}

\paragraph{Organization.} In Section~\ref{sec:prelim}, we introduce phase states, discuss separable and entangled measurements. 
In Section~\ref{sec:learning_bps}, we prove our upper and lower bounds for learning binary phase states with separable and entangled measurements. In Section~\ref{sec:learning_sparse_bps}, we prove our results on learning sparse and low-Fourier-degree phase states.  In Section~\ref{sec:generalizedphase}, we prove our upper bound for learning generalized phase states using separable and entangled measurements. Finally, in Section~\ref{sec:applications} we explicitly discuss the connection between phase states, and the diagonal unitaries in the $d$-th level of the Clifford hierarchy and $\mathsf{IQP}$ circuits. 

\paragraph{Acknowledgements.} AD was supported in part by the MIT-IBM Watson AI Lab, and in part by the U.S. Department of Energy, Office of Science, National Quantum Information Science Research Centers, Co-Design Center for Quantum Advantage under contract DE-SC0012704. AD thanks Isaac L Chuang for suggesting applications of the learning algorithms presented here and for useful comments on the draft. SA thanks Giacomo Nannicini and Chinmay Nirkhe for useful discussions. SA, SB, and TY were supported in part by the Army Research Office under Grant Number W911NF-20-1-0014.

\section{Preliminaries}
\label{sec:prelim}
\subsection{Notation.} Let $[n]=\{1,\ldots,n\}$. Let $e_i$ be an $n$-dimensional vector with $1$ in the $i$th coordinate and $0$s elsewhere. We denote the finite field with the elements $\{0,1\}$ as $\FF_2$ and the ring of integers modulo $q$ as $\ZZ_q=\{0,1,\ldots,q-1\}$ with $q$ usually being a power of $2$ in this work. For a Boolean function $f:\FF_2^n\rightarrow \FF_2$, the bias of $f$ is defined as
$$
\bias(f)=\mathop{\mathbb{E}}_{x} [(-1)^{f(x)}],
$$
where the expectation is over a uniformly random $x\in \01^n$. For $g:\FF_2^n\rightarrow \ZZ_{2^d}$, the bias of $g$ in the coordinate $j\in\FF_{2^d}^\star$ is defined as
$\bias_j(g)=\mathop{\mathbb{E}}_{x} [(\omega_{2^d})^{j\cdot g(x)}]$. For a function $f:\FF_2^n\rightarrow \FF_2$, $y\in \FF_2^{n-1}$ and $k\in [n]$, we denote 
$
(D_k f)(y)=f({y}^{k=1})+f({y}^{k=0}), 
$ where ${y}^{i=1},{y}^{i=0}\in \FF_2^n$ is defined as: the $i$th bit of $y^{i=1}$ equals $1$ and $y^{i=0}$ equals $0$ and otherwise equals $y$.

\subsection{Boolean Functions}

\paragraph{$\FF_2$ representation.} A Boolean function $f:\FF_2^n \rightarrow \FF_2$ can be uniquely represented by a polynomial over $\FF_2$ as follows:
\begin{equation}
    f(x) = \sum_{J \subseteq [n]} \alpha_J \prod_{i \in J}x_i \pmod 2,
    \label{eq:anf_boolean_function}    
\end{equation}
where $\alpha_J \in \{0,1\}$. Similar to Eq.~\eqref{eq:anf_boolean_function}, we can write Boolean functions $f:\FF_2^n \rightarrow \ZZ_{q}$ as 
\begin{equation}
    f(x) =\sum_{J \subseteq  [n]} \alpha_J  \prod_{i \in J} x_i {\pmod q}
    \label{eq:anf_boolean_function_gps}
\end{equation}
for some integer coefficients $\alpha_J \in \{0,1,\ldots ,q-1\}$. Throughout this paper, unless explicitly mentioned, we will be concerned with writing Boolean functions as a decomposition over $\FF_2$ or $\ZZ_{q}$ with $q=2^d$.  The $\FF_2$ degree of $f$ is defined as
$$
\deg(f)=\max\{|J|:\alpha_J\neq 0\}.
$$
Similarly for polynomials over $\ZZ_{2^d}$, we can define the degree as the size of the largest monomial whose coefficient $\alpha_J$ is non-negative.

We will call $g:\FF_2^n \rightarrow \FF_2$ with $g=\prod_{i \in J} x_i$ as monic monomials over $n$ variables of at most degree-$d$, characterized by set $J \subseteq [n]$, $|J| \leq d$. We will denote the set of these monic monomials by $\calM(n,d)$. Note that $|\calM(n,d)| = \sum_{j=0}^d {n \choose j} = O(n^d)$. We will denote the set of polynomials over $n$ variables of $\FF_2$-degree $d$ as $\mathcal{P}(n,d)$. Note that these polynomials are just linear combinations of monomials in $\calM(n,d)$. We will denote the set of polynomials over $n$ variables of $\FF_2$-degree $d$ with sparsity $s$ as $\mathcal{P}(n,d,s)$. Similarly, we will denote $\calP_q(n,d)$ as the set of all degree-$d$ Boolean polynomials $f:\FF_2^n \rightarrow \ZZ_{q}$ with $n$ variables. In particular, one can specify any polynomial $f\in \calP_q(n,d)$ by $O(dn^d)$ bits and $|\calP_q(n,d)|\le 2^{O(dn^d)}$.

Consider a fixed $d$, and any $x \in \FF_2^n$. Let the $d$-evaluation of $x$, denoted by $\mathrm{eval}_d(x)$, be a column vector in $\FF_2^{|\calM(n,d)|}$ with its elements being the evaluations of $x$ under different monomials $g \in \calM(n,d)$. This can be expressed as follows:
\begin{equation}
    \mathrm{eval}_d(x) = \left( \prod_{i \in J \subseteq [n], |J| \leq d} x_i \right)^\top
    \label{eq:d-evaluation_at_point_x}   
\end{equation}
For a set of points $\mathbf{x} = (x^{(1)},x^{(2)},\ldots,x^{(m)}) \in (\FF_2^n)^m$, we will call the matrix in $\FF_2^{|\calM(n,d)| \times m}$ with its $k$th column corresponding to $d$-evaluations of $x^{(k)}$, as the $d$-evaluation matrix of $\mathbf{x}$, and denote it by $Q_{\mathbf{x}}$.

\paragraph{Fourier Decomposition}
A Boolean function $f: \mathF^n \rightarrow \mathF$ admits the following Fourier decomposition
\begin{equation}
    f(x) = \sum_{J \subseteq [n]} \widehat{f}_J \chi_J(x),
    \label{eq:fourier_decomposition_polynomial_qs}
\end{equation}
where $J$ are subsets of $[n]=\{1,2,\ldots ,n\}$ and $\chi_J(x)=(-1)^{x_J}$ where $x_J=\sum_{i\in J}x_i$. Additionally the Fourier coefficients are defined as $\widehat{f}_J =\mathop{\mathbb{E}}_{x} [f(x)\chi_J(x)]$. The \emph{Fourier degree} of $f$ is defined as~$\max_{J}\{|J|:\widehat{f}_J\neq 0\}$.
Note that here all arithmetic operations use the field of real numbers $\RR$,
as opposed to the modular arithmetics used in the previous subsections. 

\subsection{Phase states}
\paragraph{Binary Phase State}   For a Boolean function $f:
\{0,1\}^n\rightarrow \{0,1\}$, we define a binary phase state as the   $n$-qubit state given by
\begin{equation}
    \ket{\psi_f} = \frac{1}{\sqrt{2^n}} \sum \limits_{x \in \Fset{n}} (-1)^{f(x)} \ket{x}.
    \label{eq:binary_phase_state}
\end{equation}
%
We use the subscript $f$ since $\ket{\psi_f}$ is characterized by $f$.

\paragraph{Generalized Phase State} We will also consider degree-$d$  generalized phase states of the form
\begin{equation}
    \ket{\psi_f} = \frac{1}{\sqrt{2^n}}\sum_{x\in \Fset{n}}\omega_q^{f(x)}\ket{x},
    \label{eq:generalized_phase_state}
\end{equation}
where $\omega_q=e^{2\pi i/q}$ and $f:\FF_2^n\rightarrow \ZZ_{q}$, with $q=2^d$, is a degree-$d$ polynomial. We consider $\ZZ_q=\{0,1,\ldots,q-1\}$ to be the ring of integers modulo $q$. 

\subsection{Useful Lemmas}

Let $e_i\in\FF_2^n$ denote the vector of all zeros except for a 1 in the $i^{\text{th}}$ coordinate.
\begin{fact}\label{lemma:existence_polynomial}
 Let $d\in [n]$, $s \leq |\calM(n,d)| = \sum_{k=1}^d {n \choose k}$, and $f \in \mathcal{P}(n,d,s)$. There exists $g_i \in \mathcal{P}(n,d-1,s)$ such that
 $   g_i(x)= f(x + e_i) + f(x) \pmod 2$  for all $x \in \{0,1\}^{n}$.
\end{fact}
The proof of this fact is straightforward. Without loss of generality, consider $i=1$. For every $f(x)=\sum_S \alpha_S \prod_{i\in S} x_i$, we can express it~as
$$
f(x)=x_1 p_1(x_2,\ldots,x_n)+p_2(x_2,\ldots,x_n),
$$
where $p_1$ has degree $\leq d-1$ and $p_2$ has degree $\leq d$. Observe that $f(x+e_1)-f(x)$ is either $p_1(x_2,\ldots,x_n)$ or $-p_1(x_2,\ldots,x_n)$ which has degree $d- 1$ and corresponds to the polynomial $g_1$ in the fact statements. This applies for every coordinate $i$. 

Note that the polynomial $g_i$ above is also often called the \textit{directional} derivative of $f$ in direction $w$ and is denoted as $D_i f$. 

\begin{fact}
\label{fact:binomial}
Let $N,s\geq 1$ such that $\gamma=s/N\leq 1/2$. Then we have
$$    
\sum \limits_{\ell=1}^{s} \binom{N}{\ell} \leq 2^{H_b(\gamma)N} \leq 2^{2\gamma \log (1/\gamma)}. 
$$
where we used above that $H_b(\gamma)={\gamma \log \frac{1}{\gamma} + (1-\gamma) \log \frac{1}{1-\gamma}}\leq 2 \gamma \log \frac{1}{\gamma}$ (for $\gamma\leq 1/2$).
\end{fact}

\begin{lemma}[The Schwartz-Zippel Lemma]\label{lem:schwartz_zippel}
Let $p(y_1, \ldots, y_n)$ be a nonzero polynomial on $n$ variables with
degree $d$. Let $S$ be a finite subset of $\mathbb{R}$, with at least $d$ elements in it. If we assign $y_1,\ldots, y_n$ values from $S$
independently and uniformly at random, then
\begin{equation}
\Pr[p(y_1, \ldots , y_n) = 0] \leq \frac{d}{|S|}.
\label{eq:schwartz_zippel_lemma_inequality}
\end{equation}
\end{lemma}

\begin{lemma}[\cite{nisan1994degree}]
\label{lem:schwartz}
Let $p(x_1,\ldots,x_n)$ be a non-zero multilinear polynomial of degree $d$. Then 
$$
\Pr_{x\in \{0,1\}^n} [p(x)=0]\leq 1-2^{-d},
$$ 
where the probability is over a uniformly random distribution on $\{0,1\}^{n}$. 
\end{lemma}

We will also need the following structural theorem about Reed-Muller codes which comments on the weight distribution of Boolean functions $f: \FF_2^n \rightarrow \FF_2$.
\begin{theorem}[{\cite[Theorem~3]{abbe2020reed}}]
\label{thm:weight}
Let $n\geq 1$ and $d\leq n/2$. Define
$
|f|=\sum_{x\in \01^n}[f(x)=1]$ and $\textsf{wt}(f)=|f|/2^n$.  
Then, for every $\varepsilon\in (0,1/2)$ and $\ell\in \{1,\ldots,d-1\}$, we have that
$$|\{f\in P(n,d):\textsf{wt}(f)\leq (1-\varepsilon)2^{-\ell}\}|\leq (1/\varepsilon)^{C\ell^4\cdot  \binom{n-\ell}{\leq d-\ell}}.
$$
Fix $w = (1 - \varepsilon)2^{n-\ell}$ and we get
$$|\{f\in P(n,d):|f|\leq w \}|\leq (1 - w/2^{n-\ell})^{-C\ell^4\cdot  \binom{n-\ell}{\leq d-\ell}}.
$$
\end{theorem}

\begin{lemma}[Fano's inequality]
\label{lem:fano}
Let $\mathsf{A}$ and $\mathsf{B}$ be classical random variables taking values in $\X$ (with $|\X|=r$) and let $q = \Pr[\mathsf{A} \neq \mathsf{B}]$. Then,
$$
    H(\mathsf{A}|\mathsf{B}) \leq  H_b(q) + q \log(r-1),
$$
where $H(\mathsf{A}|\mathsf{B})$ is the conditional entropy and $H_b(q)$ is the standard binary entropy.
\end{lemma}

\subsection{Measurements}
Throughout this paper we will be concerned with learning algorithms that use either separable or entangled measurements. Given $\ket{\psi_f}^{\otimes k}$, a learning algorithm for $f$ is said to use separable measurements if it only measure each copy of $\ket{\psi_f}$ separately in order to learn $f$. Similarly, a learning algorithm for $f$ is said to use entangled measurements if it makes an entangled measurement on the $k$-fold tensor product $\ket{\psi_f}^{\otimes k}$. In this direction, we will often use two techniques which we discuss in more detail below: sampling random partial derivatives in order to learn from separable measurements and Pretty Good Measurements in order to learn from entangled measurements. 
\subsubsection{Separable Measurements}
Below we discuss a subroutine that we will use often to learn properties about $f:\FF_2^n\rightarrow \FF_2$: given a single copy of $\ket{\psi_f}=\frac{1}{\sqrt{2^n}}\sum_{x \in \Fset{n}} (-1)^{f(x)}\ket{x}$, the subroutine produces a uniformly random $y\in \FF_2^{n-1}$ and $f(1y) + f(0y) \pmod{2}$. To this end, suppose we measure qubits $2,3,\ldots,n$ of  $\ket{\psi_f}$ in the usual $Z$ basis. We denote the resulting string as $y \in \{0,1\}^{n-1}$. The post-measurement state of qubit $1$ is then given by 
\begin{equation}
    \ket{\psi_{f,y}} = \frac{1}{\sqrt{2}} \left[ (-1)^{f(0y)}\ket{0} + (-1)^{f(1y)} \ket{1}\right].
    \label{eq:post_measurement_state}
\end{equation}
We note that $\ket{\psi_{f,y}}$ is then an $X$-basis state ($\ket{+}$ or $\ket{-}$) depending on the values of $f(1y)$ and $f(0y)$. If $f(1y) = f(0y)$, then $\ket{\psi_{f,y}}=\ket{+}$ and if $f(1y) = f(0y) + 1 \pmod 2$, then $\ket{\psi_{f,y}}=\ket{-}$. Measuring qubit $1$ in the $X$-basis and qubits $2,3,\ldots,n$ in the $Z$-basis thus produces examples of the form $(y,b)$ where $y \in \{0,1\}^{n-1}$ is uniformly random and $b = f(0y) + f(1y) \pmod{2}$. Considering Fact~\ref{lemma:existence_polynomial} with the basis of $e_1$, we note that theses examples are of the form $(y,D_{1}f(y))$, where $D_1 f(y)=f(1y) + f(0y) \pmod{2}$ is the partial derivative of $f$ along direction $e_1$. Changing the measurement basis chosen above to $ZZ\cdots X_k\cdots Z$ such that we measure all the qubits in the $Z$ basis except for the $k$th qubit which is measured in the $X$ basis, will allow us to obtain random samples of the form $(y,D_k f(y))$. Accordingly, we introduce a new subroutine.
\begin{dfn}[Random Partial Derivative Sampling (RPDS) along $e_k$] \label{def:rpds} For every $k\in[n]$, measuring every qubit of $\ket{\psi_f}$ in the $Z$ basis, except the $k$th qubit which is measured in the $X$ basis, we obtain a uniformly random $y\in \FF_2^{n-1}$ and $(D_k f)(y)$.
\end{dfn}

\subsubsection{Entangled Measurements} In general one could also consider a joint measurement applied to multiple copies of $\ket{\psi_f}$, which we refer to as entangled measurements. In this work, we will generally consider two types of entangled measurements, Bell sampling and the pretty-good measurement (which we discuss in more detail in the next section). Bell sampling is a procedure that involves measuring a quantum state (or in this case, two copies of a quantum state) in the Bell basis.\footnote{The Bell basis is the basis given by $\Big\{\frac{1}{\sqrt{2}}(\ket{00}+\ket{11}),\frac{1}{\sqrt{2}}(\ket{01}+\ket{10}),\frac{1}{\sqrt{2}}(\ket{01}-\ket{10}),\frac{1}{\sqrt{2}}(\ket{00}-\ket{11})\Big\}$.}  We will use the following version of Bell sampling that applies to the scenario where we are given noisy copies of degree-$2$ phase states (for global depolarizing noise), and covers the core idea of standard Bell sampling.  

\begin{lemma}[Bell sampling]
\label{lem:bellsampling}
Let $f:\FF_2^n\rightarrow \FF_2$ be a degree-$2$ polynomial, i.e., $f(x)=x^\top Ax$ (for an upper triangular $A\in \FF_2^{n\times n}$). Using two copies of $\rho=(1-\varepsilon)\ketbra{\psi_f}{\psi_f}+\varepsilon \id$, there exists a procedure that  outputs a uniformly random $z\in \01^n$ and $(A+A^\top)\cdot z\in \01^n$. Additionally, the same procedure, when given two copies of~$\psi_n$, outputs a uniformly random $z\in \01^n$ and $w_z \in \01^n$ that satisfies:
\[ w_z=\begin{cases} 
      (A+A^\top)\cdot z & \text{w.p. }  (1-\varepsilon)^2\\
      \text{uniformly random } v & \text{w.p. } 1-(1-\varepsilon)^2.
   \end{cases}
\]
\end{lemma}
\begin{proof}
We first consider the case where $\varepsilon=0$. Then, the following procedure produces $(z,(A+A^\top)\cdot z)$ for a uniformly random $z$. Take {two copies} of $\ket{\psi}$ 
\begin{align*}
\ket{\psi_f}\otimes \ket{\psi_f}&=\sum_{x,y\in \FF_2^n}(-1)^{x^\top Ax+y^\top Ay}\ket{x,y}\\ 
&\stackrel{\textsf{CNOT}}{\longrightarrow} \sum_{x,y}(-1)^{x^\top Ax+y^\top Ay}\ket{x,x+y} \\
&= \sum_{x,z}(-1)^{x^\top Ax+(x+z)^\top A(x+z)}\ket{x,z}=\sum_{x,z}(-1)^{x^\top(A+A^\top)z+z^\top Az}\ket{x,z}
\end{align*}
 {Measure} the second register and suppose we  {obtain $\tilde{z}$},     resulting state is
$$
(-1)^{\tilde{z}^\top A\tilde{z}}\Big(\sum_{x}(-1)^{x^\top(A+A^\top)\tilde{z}}\ket{x}\Big)\ket{\tilde{z}} \stackrel{\textsf{H}^n}{\longrightarrow}\ket{(A+A^\top)\cdot \tilde{z}}\ket{\tilde{z}},
$$   
where $\textsf{H}^n$ is the $n$-qubit Hadamard transform. Let us now consider the case where $\varepsilon >0$. Given $\rho^{\otimes 2}$, with probability $(1-\varepsilon)^2$, we obtain $(A+A^\top)\cdot \tilde{z}$ and it is not hard to see that on input $\id^{\otimes 2}$ or $\ketbra{\psi_f}{\psi_f}\otimes \id$, the output of the procedure above produces a uniformly random bit string $v\in \FF_2^n$.
\end{proof}

\paragraph{Pretty Good Measurements.}
Consider an ensemble of quantum states, $\E=\{(p_i,\ket{\psi_i})\}_{i\in[m]}$, where $p=\{p_1,\ldots,p_m\}$ is a probability distribution. In the quantum state identification problem, a learning algorithm is given an unknown quantum state $\ket{\psi_i}\in \E$ sampled according to the distribution $p$ and the learning algorithm needs to identity $i$ with probability $\geq 2/3$. In this direction, we are interested in maximizing the average probability of success to identify $i$. For a POVM specified by positive semidefinite  matrices $\M=\{M_i\}_{i\in[m]}$, the probability of obtaining outcome $j$ equals $\braketbra{\psi_i}{M_j}{\psi_i}$ and the average success probability is given by
$$
P_{\M}(\E) = \sum_{i=1}^m p_i\braketbra{\psi_i}{M_i}{\psi_i}.
$$
Let $P^{opt}(\E)=\max_{\M} P_{\M}(\E)$ denote the optimal average success probability of $\E$, where the maximization is over the set of valid $m$-outcome POVMs. For every ensemble $\E$, the so-called \emph{Pretty Good Measurement} (PGM) is a specific POVM (depending on the ensemble $\E$) that does \emph{reasonably} well against $\E$. In particular, it is well-known that 
$$
P^{opt}(\E)^2 \leq P^{PGM}(\E)\leq  P^{opt}(\E).
$$
We now define the POVM elements of the pretty-good measurement. Let $\ket{\psi'_i}=\sqrt{p_i}\ket{\psi_i}$, and $\E'=\{\ket{\psi'_i} : i\in [m]\}$ be the set of states in~$\E$, renormalized to reflect their probabilities. Define $\rho=\sum_{i\in [m]} \ketbra{\psi'_i}{\psi'_i}$. The PGM is defined as the set of measurement operators $\{\ketbra{\nu_i}{\nu_i}\}_{i\in [m]}$ where $\ket{\nu_i}=\rho^{-1/2}\ket{\psi'_i}$ (the inverse square root of $\rho$ is taken over its non-zero eigenvalues). We will use the properties of these POVM elements later on and will also need the following theorems about~PGMs.
\begin{theorem}[\cite{harrow2012many}]
\label{thm:pgm}
Let $\Sh=\{\rho_1,\ldots,\rho_m\}$. Suppose $\rho \in \Sh$ is an unknown quantum state picked from $\Sh$. Let $\max_{i\neq j}\|\sqrt{\rho_i}\sqrt{\rho_j}\|_1\leq F$. Then, given 
$$
M=O\big((\log (m/\delta)) /\log(1/F)\big)
$$
copies of $\rho$, the Pretty good measurement identifies $\rho$ with probability at least $1-\delta$.
\end{theorem}
The above theorem in fact implies the following stronger statement immediately (also stated in \cite{barnum2002reversing}) that we use here.
\begin{corollary}
\label{cor:pgm}
Let $\Sh=\{\rho_1,\ldots,\rho_m\}$. Suppose $\rho \in \Sh$ is an unknown quantum state picked uniformly from $\Sh$. Suppose there exists $k$ such that
$$
\frac{1}{m}\sum_{i\neq j}\|\sqrt{\rho_i^{\otimes k}}\sqrt{\rho_j^{\otimes k}}\|_1\leq \delta,
$$
then given $k$ copies of $\rho$, the Pretty good measurement identifies $\rho$ with probability at least $1-\delta$.
\end{corollary}

\section{Learning Binary Phase States}\label{sec:learning_bps}
In this section, we consider the problem of learning binary phase states as given by Eq.~\eqref{eq:binary_phase_state}, assuming that $f$ is a Boolean polynomial of $\FF_2$-degree $d$.

\subsection{Learning algorithm using separable measurements} \label{sec:learning_bps_separable_measurements}
We now describe our learning algorithm for learning binary phase states $\ket{\psi_f}$ when $f$ has $\FF_2$-degree $d$, using separable measurements. We carry out our algorithm in $n$ rounds, which we index by $t$. In the $t$-th round, we perform RPDS along $e_t$ (Def.~\ref{def:rpds}) in order to obtain samples of the form $(y,D_t f(y))$ where $y \in \{0,1\}^{n-1}$. For an $m\geq 1$ to be fixed later, we use RPDS on $m$ copies of $\ket{\psi_f}$ to obtain $\{\left(y^{(k)}, D_t f(y^{(k)})\right)\}_{k \in [m]}$ where $y^{(k)} \in \{0,1\}^{n-1}$ is uniformly random. We now describe how to learn $D_t f$ using these $m$ samples. 

Using Fact~\ref{lemma:existence_polynomial}, we know that $D_t f \in \calP(n-1,d-1)$. Thus, there are at most $N = |\calM(n-1,d-1)| = \sum_{k=1}^{d-1} \binom{n}{k} = n^{O(d)}$ monomials in the $\FF_2$ representation of $D_t f$. Let $A_t \in \FF_2^{m \times N}$ be the transpose of the $(d-1)$-evaluation matrix (defined in Eq.~\eqref{eq:d-evaluation_at_point_x}), such that the $k$th row of $A_t$ corresponds to the evaluations of $y^{(k)}$ under all monomials in $\calM(n-1,d-1)$, i.e., $\big(y^{(k)}_{S}\big)_{|S|\leq d-1}$, where $y^{(k)}_S=\prod_{j\in S}y^{(k)}_j$, and let $\beta_t =(\alpha_S)_{|S|\leq d-1}$ be the vector of unknown coefficients. Obtaining $\{(y^{(k)}, D_t f(y^{(k)}))\}_{k \in [m]}$, allows one to solve $A_t \beta_t = D_t f (\mathbf{y})$ for $\beta_t$ (where $\mathbf{y}=(y^{(1)},\ldots,y^{(m)})$ and $(D_t f(\mathbf{y}))_k=D_t f(y^{(k)})$) and learn the $\FF_2$-representation of $D_t f$ completely. Over $n$ rounds, one then learns $D_1 f, D_2 f,\ldots,D_n f$. The $\FF_2$-representations of these partial derivatives can then be used to learn $f$ completely, as show in Fact~\ref{fact:stitchingpoly}. This procedure is shown in Algorithm~\ref{algo:learn_bps_separable_measurements}.
\begin{fact}
\label{fact:stitchingpoly}
Let $f:\FF_2^n\rightarrow \FF_2$ be such that $f\in \calP(n,d)$. Learning $D_1 f,\ldots,D_n f$ suffices to learn $f$.
\end{fact}
\begin{proof}
Let the $\FF_2$-representation of the unknown $f$ be
\begin{equation}
    f(x) = \sum \limits_{J \subseteq [n], |J| \leq d} \alpha_J \prod_{i \in J} x_i.
\end{equation}
The $\FF_2$-representation of $D_t f$ for any $t \in \{1,2,\ldots,n\}$ is then given by
\begin{equation}
    D_t f(x) = \sum \limits_{ \substack{J \subseteq [n]:\\ t \in J, |J| \leq d} } \alpha_J \prod_{i \in J \setminus t} x_i,
\end{equation}
where we notice that $D_t f$ only contains those monomials that correspond to sets $J$ containing the component $x_t$. Let the $\FF_2$-representation of $D_t f$ with the coefficient vector $\beta_t$ be given by
\begin{equation}
    D_t f(x) = \sum_{S \subset [n], |S| \leq d-1} (\beta_t)_S \prod_{i \in S} x_i.
\end{equation}
Suppose an algorithm learns $D_1f,\ldots,D_nf$. In order to learn $f$, we must retrieve the coefficients $\alpha_J$ from the learned coefficients $\{\beta_t\}_{t \in \{1,2,\ldots,n\}}$. We accomplish this by noting that $(\beta_t)_S = \alpha_{S \cup t}$ or in other words, $\alpha_J = \{\beta_t\}_{J \setminus t}, \, t \in J$. However, there may be multiple values of $t$ that will allow us retrieve $\alpha_J$. For example, suppose $f$ contains the monomial term $x_1 x_2 x_3$ (i.e., $J=\{1,2,3\}$) then $\alpha_{\{1,2,3\}}$ could be retrieved from $(\beta_1)_{\{2,3\}}$, $(\beta_2)_{\{1,3\}}$, or $(\beta_3)_{\{1,2\}}$. When $D_t f$ (or $\beta_t$) for all $t$ is learned with zero error, all these values coincide and it doesn't matter which learned coefficient is used. When there may be error in learning $D_t f$ (or $\beta_t$), we can carry out a majority vote: $\alpha_J = \mathrm{Majority}(\{(\beta_t)_{J \setminus t} | t \in J\})$ for all $J \subseteq [n], |J| \leq d$. The majority vote is guaranteed to succeed as long as there is no error in at least half of the contributing $\beta_t$ (which is the case in our learning~algorithm).
\end{proof}

\begin{algorithm}[h]
    \caption{Learning binary phase states through separable measurements} \label{algo:learn_bps_separable_measurements}
    \textbf{Input}: Given $M=O( (2n)^d)$ copies of $\ket{\psi_f}$ where $f \in \mathcal{P}(n,d)$ 
    \begin{algorithmic}[1]
    \For{qubit $t=1,\ldots ,n$}
        \State Set $m=M/n$
        \State Perform RPDS along $e_t$ to obtain $\{(y^{(k)},D_t f(y^{(k)})\}_{k \in [m]}$ by measuring $m$ copies of $\ket{\psi_f}$.
        \State Solve the linear system of equations $A_t\cdot \beta_t=D_t f (\mathbf{y})$ to learn $D_t f$ explicitly.
    \EndFor 
    \State Use Fact~\ref{fact:stitchingpoly} to learn $f$ using $D_1 f,\ldots,D_n f$ (let $\tilde{f}$ be the output).
    \end{algorithmic}
    \textbf{Output}: Output $\tilde{f}$
\end{algorithm}
We now prove the correctness of this algorithm.
\begin{theorem}
\label{thm:binaryphasedegreed}
Let $n\geq 2,d\leq n/2$. Algorithm~\ref{algo:learn_bps_separable_measurements} uses  $M=O(2^d n^{d})$ copies of an unknown $\ket{\psi_f}$ for $f\in \calP(n,d)$ and with  high probability identifies $f$ using single qubit $X,Z$ measurements. \end{theorem}
\begin{proof}
Algorithm~\ref{algo:learn_bps_separable_measurements} learns $f$ by learning $D_1 f,\ldots,D_n f$ and thereby learns $f$ completely. Here we prove that each $D_t f$ can be learned with $m=O(2^d n^{d-1})$ copies of $\ket{\psi_f}$ and an exponentially small probability of error. This results in an overall sample complexity of $O(2^d n^d)$ for learning $f$ and hence $\ket{\psi_f}$.  Let us consider round $t$ in Algorithm~\ref{algo:learn_bps_separable_measurements}. We generate $m$ constraints $\{\left(y^k,(D_t f)(y^{(k})\right)\}_{k \in [m]}$ where $y^{(k)} \in \FF_2^{n-1}$ by carrying out RPDS along $e_t$ on $m$ copies of $\ket{\psi_f}$.

We learn the $\FF_2$-representation of $D_t f$ by setting up a linear system of equations using these $m$ samples: $A_t \beta_t = D_t f(\mathbf{y})$, where $A_t$ is the transposed $(d-1)$-evaluation matrix in round $t$, evaluated over $\mathbf{y} = (y^{(1)}, y^{(2)},\ldots y^{(m)})$, and $\beta_t \in \FF_2^{|\calM(n-1,d-1)|}$ is the collective vector of coefficients corresponding to the monomials in $\calM(n-1,d-1)$. By construction, this system has at least one solution. If there is exactly one solution, then we are done. Otherwise, the corresponding system has a non-zero solution, that is, there exists a non-zero degree-$(d-1)$ polynomial $g\, : \, \FF_2^{n-1} \to \FF_2$ such that $g(y^{(j)})=0$ for all $j=1,2,\ldots,m$. 

Below we prove that the probability of this bad event can be bounded through the Schwartz-Zippel lemma. Applying Lemma~\ref{lem:schwartz} and by noting that $y^j \in \FF_2^{(n-1)}$ are independent and uniformly distributed, we have that
\begin{equation}
    \Pr[g(y^{(1)}) = g(y^{(2)}) = \cdots = g(y^{(m)}) = 0] \leq (1-2^{-d})^{m}\leq e^{-m2^{-d}}
    \label{eq:schwartz_zippel_lemma_multiple_samples}
\end{equation}
Let $\mathcal{P}_{\text{nnz}}(n,d)$ be the set of all degree-$d$ polynomials $g:\FF_2^n \rightarrow \FF_2$ which are not identically zero. Define event 
\be
\mathsf{BAD}(y^1,\ldots,y^m) =[\exists g \in \mathcal{P}_{\text{nnz}}(n-1,d-1) \, : \, g(y^1)=\ldots=g(y^m)=0 {\pmod 2}].
\ee
We note that $|\mathcal{P}_{\text{nnz}}(n-1,d-1)| \leq 2^N$ where $N = O(n^{d-1})$. By union bound and Eq.~\eqref{eq:schwartz_zippel_lemma_multiple_samples}, we have %
\begin{equation}
    \Pr[\mathsf{BAD}(y^{(1)},\ldots,y^{(m)})] \leq |\mathcal{P}_{\text{nnz}}(n-1,d-1)| \cdot (1-2^{-d})^m \leq 2^{n^{d-1} - m2^{-d}(\ln 2)}.
\end{equation}
Thus choosing $m=O((2n)^{d-1})$ is enough to learn all coefficients $\{\alpha_J\}_{t \in J}$ (through $\beta_t$) in the $\FF_2$ representation of $f$ with an exponentially small probability of error. We need to repeat this over all the $n$ qubits in order to learn $D_1 f,\ldots, D_n f$ and then use Fact~\ref{fact:stitchingpoly}  to learn $f$ completely. This gives an overall sample complexity of $O((2n)^d)$ for learning binary phase states. Observe that the only measurements that we needed in this algorithm were single qubit $\{X,Z\}$ measurements. 
\end{proof}

\begin{corollary}
An $n$-qubit state $\ket{\psi_f}$ with the unknown Boolean function $f$ of given Fourier-sparsity $s$ can be learned with Algorithm~\ref{algo:learn_bps_separable_measurements} that consumes $M$ copies of $\ket{\psi_f}$ with probability $1-2^{-\Omega(n)}$ provided that $M \geq O(s n^{\log s})$.
\end{corollary}
 The proof of this corollary simply follows from the following: for a Boolean function, the Fourier sparsity $s$ of $f$ is related to the $\FF_2$-degree $d$ of $f$ \cite{bernasconi1999spectral} as
$d \leq \log s$. Along with Theorem~\ref{thm:binaryphasedegreed} we obtain the corollary.

\subsection{Learning using entangled measurements} \label{sec:learning_bps_entangled_measurements}
We now consider the problem of learning binary phase states using entangled measurements. We have the following result.

\begin{theorem}
\label{thm:entangledupperbound}
Let $n\geq 2,d\leq n/2$. There exists an algorithm that uses  $M=O((2n)^{d-1})$ copies of an unknown $\ket{\psi_f}$ for $f\in \calP(n,d)$ and identifies $f$ using entangled measurements with probability~$\geq 2/3$. There is also a lower bound of $\Omega(n^{d-1})$ for learning these states.
\end{theorem}

\begin{proof}
In order to prove this theorem, we follow the following  steps. We first observe that the optimal measurement for our state distinguishing problem is the pretty good measurement (PGM). Second we observe that the success probability of the PGM is the same for \emph{every} concept in the ensemble. We  bound the success probability of the PGM  using Corollary~\ref{cor:pgm} we get our upper~bound.

For $f\in \calP(n,d)$, let $U_f$  be the unitary defined as $U_f=\textsf{diag}\big(\{(-1)^{f(x)}\}_x\big)$, that satisfies $U_f \ket{+}^n=\ket{\psi_f}$. Observe that the set $\{U_f\}_{f\in \calP(n,d)}$ is an Abelian group.  The ensemble we are interested in is  $\Sh=\{U_f \ket{+}^n\}_{f\in \calP(n,d)}$ and  such an ensemble is called \emph{geometrically uniform} if the $\{U_f\}$ is an Abelian group. A well-known result of Eldar and Forney~\cite{eldar2001quantum} showed that the optimal measurement for state distinguishing a geometrically uniform (in particular $\Sh$) is  the pretty-good~measurement. We now show that the success probability of the PGM is the same for every state in the ensemble. In this direction, for $M\geq 1$, let $\sigma_f=\ketbra{\psi_f}{\psi_f}^{\otimes M}$. The POVM elements of the pretty good measurement $\{E_f:f\in \calP(n,d)\}$ is given by the POVM elements $E_f=S^{-1/2}\sigma_fS^{-1/2}$ where $S=\sum_{f\in \calP(n,d)}\sigma_f$. The probability that the PGM identifies the unknown $\sigma_f$ is given by
$$
\Pr(f)=\Tr(\sigma_f E_f)=\langle \psi_f^{\otimes M}|S^{-1/2}|\psi_f^{\otimes M}\rangle^2.
$$
Our claim is that $\Pr(f)$ is the same for every $f\in \calP(n,d)$. Using the Abelian property of the unitaries $\{U_f\}_f$, observe that $U_f\ket{\psi_g}=\ket{\psi_{f\oplus g}}$ for every $f,g\in \calP(n,d)$. Thus, we have that 
$
(U_f^{\otimes M})^{\dagger} SU_{f}^{\otimes M}=S,
$
which implies that$
(U_f^{\otimes M})^{\dagger} S^{-1/2}U_{f}^{\otimes M}=S^{-1/2}.
$
Hence it follows~that
$$
\Pr(f)=\big(\langle +|^{\otimes M}(U_f^{\otimes M})^{\dagger} S^{-1/2}U_{f}^{\otimes M}|+\rangle^{\otimes M}\big)^2= \big(\langle +|^{\otimes M} S^{-1/2}|+\rangle^{\otimes M}\big)^2=\Pr(0),
$$
for every $f\in \calP(n,d)$. Finally,  observe that
$
\langle \psi_f |\psi_g\rangle =\mathop{\mathbb{E}}_{x}\big[(-1)^{f(x)+g(x)}\big]=1-2\Pr_x [f(x)\neq g(x)].
$
Let $\calP^*(n,d)$ be the set of non-constant polynomials in $\calP(n,d)$. We now have the following

\begin{align*}
\frac{1}{2^{\binom{n}{\leq d}}}\sum_{\substack{f\neq g:\\f,g\in P(n,d)}}\|\sqrt{\rho_f^{\otimes k}}\sqrt{\rho_g^{\otimes k}}\|_1 &=\sum_{\substack{g\in P^*(n,d)} } \big(1-2\Pr_x[g(x)=1]\big)^{2k}\\
&=\sum_{\substack{g\in P^*(n,d)} } \big(1-2\textsf{wt}(g)\big)^{2k}\\
&=\sum_{\ell=1}^{d-1}\sum_{\substack{g\in P^*(n,d)}} \big(1-2|g|/2^n\big)^{2k}\cdot \Big[|g|\in [2^{n-\ell-1},2^{n-\ell}-1]\Big]\\
&=\sum_{\substack{g\in P^*(n,d)}} \big(1-2|g|/2^n\big)^{2k}\cdot \Big[|g|\in [2^{n-2},2^{n-1}-1]\Big] \\ 
& \hspace{4em} + \sum_{\ell=2}^{d-1}\sum_{\substack{g\in P^*(n,d)}} \big(1-2|g|/2^n\big)^{2k}\cdot \Big[|g|\in [2^{n-\ell-1},2^{n-\ell}-1]\Big]\\
&\leq 2^{n-1}2^{-2k+C\binom{n-1}{\leq d-1}}+\sum_{\ell=2}^{d-1}\big(1-\frac{1}{2^\ell}\big)^{2k}\sum_{\substack{g\in P^*(n,d)}}  \Big[|g| \leq 2^{n-\ell}\Big],
\end{align*}
where the first equality used that the PGM has the same success probability for every $f,g\in \calP(n,d)$, third equality used that $|g|\ge 2^{n-d}$ for any non-zero polynomial $g\in P(n,d)$~\cite{macwilliams1977theory} and last inequality used Theorem~\ref{thm:weight}. For $k=O(n^{d-1})$ (by picking a sufficiently large constant in $O(\cdot)$), the first term is at most $\leq 1/100$.  To bound the second term, using Theorem~\ref{thm:weight} we have
$$
\sum_{\ell=2}^{d-1}\big(1-\frac{1}{2^\ell}\big)^{2k}\sum_{\substack{g\in P^*(n,d)}}  \Big[|g| \leq 2^{n-\ell}\Big]\leq \sum_{\ell=2}^{d-1}2^{n-\ell}\exp(-2k/2^{\ell}+(n-\ell)\ell^4\binom{n-\ell}{\leq d-\ell}).
$$
Each term is $\exp(-n^{d-1})$ for $k=O(n^{d-1})$, so the overall sum is $\leq 1/100$.
Corollary~\ref{cor:pgm} implies our desired upper bound.

In order to see the lower bound, observe that each state $\ket{\psi_f}$ contains $n$ bits of information and the goal of the learning algorithm is to learn an unknown $f$, i.e., obtain $O(n^d)$ bits of information. Hence by Holevo's theorem~\cite{holevo1973bounds}, one requires $\Omega(n^{d-1})$ copies of the unknown state for state identification.\footnote{We refer the reader to Montanaro~\cite[Proposition~1]{montanaro2012quantum} for a detailed exposition of this lower bound proof.}
\end{proof}

\subsection{Lower bounds}
In the last section we saw that $\Theta(n^{d-1})$ many copies of $\ket{\psi_f}$ with degree-$d$ are necessary and sufficient to learn $f$ if we allowed only entangled measurements. Earlier we saw that $O(n^d)$ many copies of $\ket{\psi_f}$ sufficed to learn $f$ using separable measurements. A natural question is: Can we learn $f$ using fewer copies if we are restricted to using only separable measurements? In the theorem below, we provide a lower bound that complements our upper bound, thereby showing $\Theta(n^d)$ copies are necessary and sufficient to learn $f$ using separable measurements.

\begin{theorem}
\label{thm:lowerboundbinaryphase}
Let $2\leq d\leq n/2$. Suppose there exists an algorithm that with probability $\geq 1/10$, learns an $n$-variate
polynomial $f\in \calP(n,d)$, given $M$ copies of the phase state $\ket{\psi_f}=\frac{1}{\sqrt{2^n}}\sum_{x \in \Fset{n}} (-1)^{f(x)}\ket{x}$, measuring each copy in  an arbitrary  orthonormal basis,  and performing an arbitrary classical processing.~Then  
\be
M=\Omega(\log |\calP(n,d)|)=\Omega(n^d).
\ee
\end{theorem}
\begin{proof}
Let $U$ be an $n$-qubit unitary operator. Define the probability distribution
\be
\Pr_{U}(x|f)=|\la x|U|\psi_f \ra|^2.
\ee
Let $H_{U}(x|f)$ be the Shannon entropy of $x$ sampled from $\Pr_{U}(x|f)$, i.e., 
\be
H_{U}(x|f) =-\sum_x \Pr_{U}(x|f) \log{\Pr_{U}(x|f)}.
\ee
Below we prove the following.
\begin{lemma}
\label{lemma:entropy}
Suppose $d\ge 2$ and $U$ is an $n$-qubit unitary. Then
\be
\label{entropy_bound}
\EE_f  [H_{U}(x|f)] \ge n-2,
\ee
where the expectation is over a uniformly random $f\in \calP(n,d)$
\end{lemma}
We will assume the lemma now and prove the theorem statement. Below we assume that $f\in \calP(n,d) $ is picked uniformly at random.
Suppose we measure the $j$-th copy of $|\psi_f \ra$
in a basis $\{ U_j^\dag |x\ra\}_x$ for some $n$-qubit unitary $U_j$.
Let $x^1,x^2,\ldots,x^M\in \FF_2^n$ be the measured bit strings. 
The joint probability  distribution of $f$ and $x$ is given by
\be
\Pr(f,x) = \frac1{| \calP(n,d) |} \prod_{j=1}^M \Pr_{U_j}(x^j|f).
\ee
The conditional entropy of $x$ given $f$ is
\be
H(x|f) =
\frac1{| \calP(n,d) |} \sum_{f\in  \calP(n,d) } \; \sum_{j=1}^M H_{U_j}(x^j|f)
=\sum_{j=1}^M \EE_f  H_{U_j}(x^j|f) \ge M(n-2),
\ee
where the inequality used Lemma~\ref{lemma:entropy}.
It follows that the conditional entropy of $f$ given $x$ obeys 
\be
\label{Hfx}
H(f|x) = H(x|f) -H(x) + H(f) \ge M(n-2)- H(x)+H(f).
\ee
Since $H(x)\le n M$ and $H(f)=\log |\calP(n,d)|$, we get
\be
H(f|x) \ge \log |\calP(n,d)| - 2M.
\ee
Assuming there exists a learning algorithm that given $x$ learns $f$ with probability $\geq 1/10$, by Fano's inequality (Lemma~\ref{lem:fano}), we know that $H(f|x)\leq H_b(1/10)+(1/10)\cdot \log |\calP(n,d)|$. It remains to prove Lemma~\ref{lemma:entropy}. 
\begin{proof}[\bf Proof of Lemma~\ref{lemma:entropy}]
It is  known that the Shannon entropy of a distribution is lower bounded by its
Renyi entropy of order two. Thus we have 
\be
\label{H1H2}
H_U(x|f)\ge -
\log{\left[ R_U(x|f)  \right]},
\ee
where
\be
R_U(x|f):= \sum_x (\Pr_{U}(x|f))^2 = \sum_x \la x,x|  U^{\otimes 2}|\psi_f \ra\la \psi_f|^{\otimes 2} (U^\dag)^{\otimes 2} |x,x\ra.
\ee
Taking the expected value of Eq.~(\ref{H1H2}) and noting that  the function $-\log{(\cdot)}$ is convex, one gets
\be
\label{EHbound}
\EE_f  [H_U(x|f)]  \ge -\EE_f\left[\log{ R_U(x|f)}  \right] \ge -\log{ \EE_f\left[ R_U(x|f) \right]}.
\ee
Below we prove 
\begin{prop}
\label{prop:1}
Let $f \in \calP(n,d) $ be a  uniformly random degree-$d$ polynomial with $d\ge 2$.
Then 
\be
\label{second_moment}
\EE_f [|\psi_f \ra\la \psi_f|^{\otimes 2}] = 
\frac1{4^n} (\id+\mathrm{SWAP}) + \frac1{2^n} |\Phi^+\ra\la \Phi^+| - \frac2{4^n} \sum_x |x,x\ra\la x,x|,
\ee
where $\mathrm{SWAP}$ swaps two $n$-qubit registers and
$|\Phi^+\ra = 2^{-n/2} \sum_x |x,x\ra$
is the EPR state of $2n$ qubits.
\end{prop}
Combining the proposition and the bound $|\la \Phi^+|U^{\otimes 2}|x,x\ra|^2\le 2^{-n}$
gives 
\be
\EE_f [R_U(x|f)]  \le \frac{3}{2^n}  \le \frac1{2^{n-2}}.
\ee
Substituting this into 
Eq.~(\ref{EHbound}) completes the proof. We now prove the proposition. 
\begin{proof}[\bf Proof of Proposition~\ref{prop:1}]
Let
\[
Q=\EE_f \big[|\psi_f \ra\la \psi_f|^{\otimes 2}\big] =\frac1{4^n} \sum_{w,x,y,z} E(w,x,y,z) |w,x\ra\la y,z|,
\]
where
\[
E(w,x,y,z)=\EE_f \big[(-1)^{f(w) + f(x)+f(y)+f(z)}\big].
\]
Our proof strategy uses a couple of lemmas from~\cite[Proposition~5]{bravyi2019simulation} and~\cite[Lemma~11]{bremner2016average}.
\begin{claim}
$E(w,x,y,z)=0$ unless $w+x+y+z =0^n$ and at least two of the strings $w,x,y$ coincide.
\end{claim}
\begin{proof}
We can write
\[
f(v) = \sum_{p=1}^n A_p v_p + \sum_{1\le p<q\le n} A_{p,q} v_p v_q + \ldots {\pmod 2}
\]
where $A_p\in \{0,1\}$ and $A_{p,q}\in \{0,1\}$ are picked uniformly at random 
and dots represents higher order terms.
Taking the expectation value over $A_p$ gives 
\[
\EE_{A_p} \Big[(-1)^{A_p(w_p + x_p + y_p + z_p)}\Big] = 0 \quad \mbox{unless} \quad w_p+ x_p + y_p + z_p = 0 {\pmod 2}.
\]
This proves the first part of the claim. Taking the expectation value over $A_{p,q}$ gives
\[
\EE_{A_{p,q}} \Big[(-1)^{A_{p,q}(w_pw_q + x_px_q + y_py_q + z_p z_q)}\Big] = 0
\quad \mbox{unless} \quad w_pw_q+ x_px_q + y_py_q + z_pz_q = 0 {\pmod 2}.
\]
Substituting $z_p=x_p+y_p+w_p {\pmod 2}$ and $z_q=x_q+y_q+w_q {\pmod 2}$ in the above
expression one concludes that $E(w,x,y,z)=0$ unless $w+x+y+z =0^n{\pmod 2}$ and
\be
\label{eq1}
(w_p x_q + x_p w_q ) + (x_p y_q + y_p x_q)  + (w_p y_q + y_p w_q) = 0 {\pmod 2}.
\ee
If $w=x=y$ then we are done. Otherwise there exists an index $p\in [n]$
such that exactly two of the variables $w_p$, $x_p$, $y_p$ coincide. 
Since Eq.~(\ref{eq1}) is symmetric under permutations of $w,x,y$ we can assume wlog that
$x_p=y_p\ne w_p$. Consider two cases:\\
{\em Case~1:} $x_p=y_p=0$ and $w_p=1$. Then Eq.~(\ref{eq1}) gives $x_q=y_q$ for all $q\ne p$. Thus $x=y$.\\
{\em Case~2:} $x_p=y_p=1$ and $w_p=0$. Then Eq.~(\ref{eq1}) gives  $w_q+y_q+x_q+w_q=0{\pmod 2}$ for all $q\ne p$.
Thus $x_q=y_q$ for all $q\ne p$, that is, $x=y$.
\end{proof}
Note  that $E(w,x,y,z)=1$ whenever $w+x+y+z =0^n{\pmod 2}$ and at least two of the strings $w,x,y$ coincide.
For example, if $w=x$ then one must have $y=z$ and thus the sum $f(w)+f(x)+f(y)+f(z)$ is zero modulo two.
This leads to
\[
Q=\frac1{4^n} \sum_{w,x,y,z} E(w,x,y,z) |w,x\ra\la y,z|
=\frac1{4^n} (I+\mathrm{SWAP}) + \frac1{2^n} |\Phi^+\ra\la \Phi^+| - \frac2{4^n} \sum_x |x,x\ra\la x,x|.
\]
Here the last term is introduced to avoid overcounting. 
\end{proof}
This concludes the proof of the proposition and Lemma~\ref{lemma:entropy}. 
\end{proof}
The proof of this lemma concludes the proof of Theorem~\ref{thm:lowerboundbinaryphase}. 
\end{proof}

\section{Learning sparse and low Fourier-degree binary phase states}\label{sec:learning_sparse_bps}
In this section, we first consider the problem of learning binary phase states (Eq.~\eqref{eq:binary_phase_state}) using separable measurements under the assumption that $f$ is an $s$-sparse $\FF_2$-degree $d$ Boolean function written as
\begin{equation}
    f(x) = \sum \limits_{J\subseteq [n]} \alpha_{J} \prod_{i \in J} x_i \pmod 2
    \label{eq:anf_sparse_polynomial_qs}
\end{equation}
where $|\{J:\alpha_J\neq 0\}|=s$, i.e., there are $s$ terms in the $\FF_2$ representation of $f$. 

\subsection{Sparse Learning Algorithm}
Our algorithm for learning sparse binary phase states and analysis of its sample complexity is similar to that in Section~\ref{sec:learning_bps_separable_measurements}. Similar to Algorithm~\ref{algo:learn_bps_separable_measurements} in the $t$-th round, for an $m \geq 1$ to be fixed later (where $m$ is the sample complexity), we use RPDS along $e_t$ on $m$ copies of $\ket{\psi_f}$ to obtain $m$ samples $\{\left(y^{(k)}, D_t f(y^{(k)})\right)\}_{k \in [m]}$ where $y^{(k)} \in \{0,1\}^{n-1}$ is uniformly random. We now describe how to learn $D_t f$ using these $m$ examples.

Let $A_t \in \FF_2^{m\times |\calM(n-1,d-1)|}$ be the transposed $(d-1)$-evaluation matrix (defined in Eq.~\eqref{eq:d-evaluation_at_point_x}) such that the $k$th row of $A_t$ is given by the vector $\big(y^{(k)}_{S}\big)_{|S|\leq d-1}$, where $y^{(k)}_S=\prod_{j\in S}y^{(k)}_j$. We can then write a system of linear equations $A_t \beta_t = D_t f(\mathbf{y})$ where $\beta_t \in \FF_2^{|\calM(n-1,d-1)|}$ is the vector of coefficients corresponding to monomials in $\calM(n-1,d-1)$ and $(D_t f(\mathbf{y}))_k = D_t f(y^{(k)})$. Instead of explicitly solving $A_t \beta_t = D_t f(\mathbf{y})$ (which we did in Section~\ref{sec:learning_bps_separable_measurements}), we propose to estimate the unknown coefficients vector $\beta_t$ by solving the following linear program by drawing connection to compressed sensing~\cite{draper2009compressed}.
\begin{equation}
    \hat{\beta}_t \in \argmin \norm{\beta}_1 \text{ such that }  A_t \beta = D_t f(\mathbf{y}) {\pmod 2} \quad \text{ over } \beta \in \FF_2^{|\calM(n-1,d-1)|}.
    \label{eq:optimization_learning_sparse_bps}
\end{equation}
The solution vector $\hat{\beta}_t$ produced by the above linear program corresponds to solving the subset of coefficients $\alpha_J$ in the $\FF_2$-representation of $f$ corresponding to sets $J$ which contain $t$. Like in Section~\ref{sec:learning_bps_separable_measurements}, we repeat the above procedure over $n$ rounds to learn $D_1 f, D_2 f,\ldots, D_n f$ and then eventually learn $f$ using Fact~\ref{fact:stitchingpoly}.  We give details of this algorithm in Algorithm~\ref{algo:learn_sbps_separable_measurements}. 
\begin{algorithm}[h]
    \caption{Learning sparse binary phase states through separable measurements} \label{algo:learn_sbps_separable_measurements}
    \textbf{Input}: Access to $M = O(2^d s d n \log n)$ copies of $\ket{\psi_f}$ where $f \in \mathcal{P}(n,d,s)$
    \begin{algorithmic}[1]
    \For{qubit $t=1,\ldots,n$}
        \State Set $m=M/n$
        \State Perform RPDS to obtain $\{(y^{(k)},D_t f(y^{(k)})\}_{k \in [m]}$ by measuring $m$ copies of $\ket{\psi_f}$.
        \State Solve linear program $\hat{\beta}_t \in \argmin \norm{\beta}_1 \text{ s.t. } A_t \beta = D_t f(\mathbf{y}) \text{ for } \beta \in \mathF^{|\calM(n-1,d-1)|}$ (Eq.~\eqref{eq:optimization_learning_sparse_bps})
    \EndFor 
    \State Use Fact~\ref{fact:stitchingpoly} to learn $f$ using $D_1 f,\ldots,D_n f$ (let $\tilde{f}$ be the output).
    \end{algorithmic}
    \textbf{Output}: Output $\tilde{f}$
\end{algorithm}

We now argue the correctness of the algorithm. 
\begin{theorem}\label{thm:sample_complexity_sparse_degree_d_polynomial_qs}
An $n$-qubit state $\ket{\psi_f}$ with the unknown Boolean function $f$ of given $\FF_2$-degree $d$ and $\FF_2$-sparsity $s \leq |\calM(n-1,d-1)|/2$ can be learned with an algorithm that consumes $M$ copies of $\ket{\psi}$ with probability $1 - 2^{- \Omega(n)}$ provided that $M = O(2^d s d n \log n)$.  Moreover the algorithm only uses $\{X,Z\}$ single-qubit measurements.
\end{theorem}
\begin{proof}
Algorithm~\ref{algo:learn_sbps_separable_measurements} learns $f$ by learning $D_1 f,\ldots,D_n f$ and thereby learns $f$ completely. Here we prove that each $D_t f$ can be learned with with $m=O(2^d s d \log n)$ copies of $\ket{\psi_f}$ and an exponentially small probability of error. This results in an overall sample complexity of $O(2^d s d n \log n)$ for learning $f$ and hence $\ket{\psi_f}$. Let us consider round $t$ in Algorithm~\ref{algo:learn_sbps_separable_measurements}. We generate $m$ samples $\{y^{(k)},D_t f(y^{(k)}\}$ through RPDS. Using these $m$ samples, we can solve Eq.~\eqref{eq:optimization_learning_sparse_bps} and obtain the solution $\hat{\beta}_t$. An error occurs when $\hat{\beta}_t \neq \beta_t^\star$ where we have denoted the true solution by $\beta^\star_t$.\footnote{Note that the true coefficients $(\beta^\star_t)_S = \alpha_{S \cup t}$ where $\alpha_J$ is the true coefficient in the $\FF_2$-representation of $f$, corresponding to set $J$.} Probability of this error occurring is then given by
\begin{align}
    \Pr[\hat{\beta}_t \neq \beta^\star_t] &= \Pr[\exists \beta \in \{0,1\}^N, \beta \neq \beta^\star_t \, \mid \, A_t \beta = D_t f(\mathbf{y}) \cap \norm{\beta}_1 \leq \norm{\beta^\star_t}_1]
\end{align}
where we have denoted $N=|\calM(n-1,d-1)|$. Below we prove that the probability of this bad event can be bounded through Schwartz-Zippel (Lemma~\ref{lem:schwartz_zippel}). Define event
\begin{align}
\mathsf{BAD}(y^{(1)},\ldots,y^{(m)}) =[\exists \beta \in \{0,1\}^N, \beta \neq \beta^\star_t \, \mid \, A_t \beta = D_t f(\mathbf{y}) \cap \norm{\beta}_1 \leq \norm{\beta^\star_t}_1].
\end{align}
Let us consider the $k$th row of $A_t$. We note that the corresponding equation can be rewritten as
\begin{align}
    (A_t)_k \cdot \beta &= (D_t f(\mathbf{y}))_k \pmod 2 \\
    (A_t)_k \cdot \beta &= (A_t)_k \beta_t^\star \pmod 2 \\
    (A_t)_k \cdot (\beta - \beta_t^\star)  &= 0 \pmod 2
\end{align} 
As $\beta \neq \beta_t^\star$, this means there exists a non-zero polynomial $g \in \calP(n-1,d-1)$ corresponding to the coefficients $\beta - \beta_t^\star \pmod 2$. Applying Lemma~\ref{lem:schwartz} and by noting that $y^{(j)} \in \FF_2^{(n-1)}$ are independent and uniformly distributed, we have that
\begin{equation}
    \Pr[g(y^{(1)}) = g(y^{(2)}) = \cdots = g(y^{(m)}) = 0] \leq (1-2^{-d})^{m}\leq e^{-m2^{-d}}
    \label{eq:schwartz_zippel_lemma_multiple_samples2}
\end{equation}
Let $\mathcal{P}_{\text{nnz}}(n,d,s)$ be the set of all degree-$d$ polynomials $g:\FF_2^n \rightarrow \FF_2$ with sparsity $s$ which are not identically zero. By union bound and Eq.~\eqref{eq:schwartz_zippel_lemma_multiple_samples2}, we have that
\begin{align}
    \Pr[\mathsf{BAD}(y^1,\ldots,y^m)] &\leq |\mathcal{P}_{\text{nnz}}(n-1,d-1,s)| \cdot (1-2^{-d})^m \\
    &= \sum \limits_{ \substack{\{\beta \in \{0,1\}^N, \beta \neq \beta^\star_t : \\ \norm{\beta}_1 \leq \norm{\beta^\star_t}_1\} }} (1-2^{-d})^m \\
    &\leq \sum \limits_{\ell=1}^{\norm{\beta^\star_t}_1} |\{\beta \in \{0,1\}^N \, : \, \norm{\beta}_1 = \ell \}| 2^{-m 2^{-d} (\ln 2)} \\
    &= \sum \limits_{\ell=1}^{s} \binom{N}{\ell} 2^{-m 2^{-d} (\ln 2)}\leq  2^{2 s 
    \log(N/s) - m 2^{-d} (\ln 2)},
\end{align}
where the final inequality used Fact~\ref{fact:binomial}. We can thus learn all the coefficients $\beta^\star_t$ with an exponentially small probability of error by choosing $m=O(2^d s \log N)$. We need to repeat this over all the $n$ qubits, giving an overall sample complexity of $O(2^d s d n \log n)$ (by noting that $N=\calM(n-1,d-1)=O(n^d)$) of learning sparse binary phase states using only separable measurements. Using Fact~\ref{fact:stitchingpoly}, we can completely learn $f$. Observe that Algorithm~\ref{algo:learn_sbps_separable_measurements} uses only single qubit $\{X,Z\}$ measurements. 
\end{proof}

\subsection{Learning low Fourier-degree phase states}
We conclude this section with a theorem about learning low Fourier-degree phase states. 
\begin{theorem}
\label{thm:fourierbinaryphasedegreed}
Consider binary phase states $\ket{\psi}=\frac{1}{\sqrt{2^n}}\sum_x(-1)^{f(x)}\ket{x}$ where $f:\01^n\rightarrow \01$ has Fourier-degree $d$. Then $\widetilde{O}(2^{2d})$ copies of $\ket{\psi}$ are sufficient to identify $f$ with probability $\geq 2/3$. 
\end{theorem}

\begin{proof}
The proof is a simple application of couple of results. Observe that the number of degree-$d$ Boolean functions is $2^{2d2^d}$. {To see this, observe that~\cite{o2014analysis} for a degree-$d$ Boolean function, all the Fourier coefficients are integer multiple of $2^{-d}$, and since by Parseval's theorem $\sum_S \widehat{f}(S)^2\leq 1$, the number of non-zero Fourier coefficients is $\leq 2^{2d}$. Hence the number of degree-$d$ Boolean functions is $(2^{d})^{2^{2d}}=2^{2d2^d}$.} Additionally,  for degree-$d$ functions $f,g$ we have that 
$$
\langle \psi_f | \psi_g\rangle =\mathop{\mathbb{E}}_x[(-1)^{f(x)+g(x)}]=1-2\Pr_{x}[f(x)\neq g(x)]\leq 1-2^{-d+1}
$$
where the final inequality uses the Schwartz-Zippel lemma (Lemma~\ref{lem:schwartz}) and the fact that $f-g$ has degree at most $d$. Now, putting this together with Theorem~1 in~\cite{montanaro2019pretty}, we get that the number of copies of $\ket{\psi_f}$ sufficient to learn $f$ is given by
$
O\big(\big(\log 2^{2d2^d}\big)/2^{-d}\big)=\widetilde{O}(2^{2d}).
$
\end{proof}

\section{Learning generalized phase states}
\label{sec:generalizedphase}
In this section, we consider the problem of learning generalized phase states $\ket{\psi_f}$ 
as given by Eq.~(\ref{eq:generalized_phase_state}), assuming that $f$ is a 
degree-$d$ $\ZZ_q$-valued polynomial,
 $f\in \calP_q(n,d)$. Note that since our goal is to learn $\ket{\psi_f}$ up to an overall phase, we shall identify polynomials which differ only by a constant~shift.
\begin{dfn}
Polynomials $f,g\in \calP_q(n,d)$ are  equivalent if $f(x)-g(x)$ is a constant function.
\end{dfn}
To simplify notation, here and below we omit modulo operations keeping in mind that degree-$d$ polynomials take values in the ring $\ZZ_q$. Thus all equal or not-equal 
constraints  that involve a polynomial's value are modulo $q$.

\subsection{Learning using separable measurements}

Let $q\ge 2$ and $d\ge 1$ be  integers.
For technical reasons, we shall assume that $q$ is even.
Let $\omega_q=e^{2\pi i/q}$.
Our main result is as follows.
\begin{theorem}
\label{thm:generalizedphasedegreed}
Let $d\leq n/2$. There exists an  algorithm that uses $M=O(2^d q^3 n^{d}\log{q})=O(n^d)$ copies of a generalized phase state $\ket{\psi_f}=\frac{1}{\sqrt{2^n}}\sum_{x\in \{0,1\}^n} \omega_q^{f(x)}\ket{x}$ with an unknown polynomial $f\in \calP_q(n,d)$  and  outputs a polynomial $g\in \calP_q(n,d)$ such that  $g$ is equivalent to $f$ with the probability at least $1-2^{-\Omega(n)}$. The quantum part of the algorithm requires only single-qubit unitary gates and measurements in the standard basis.

Moreover, suppose there exists an algorithm that with probability $\geq 1/10$, learns an $n$-variate
polynomial $f\in \calP_q(n,d)$, given $k$ copies of the phase state $\ket{\psi_f}$, measuring each copy in  an arbitrary  orthonormal basis,  and performing an arbitrary classical processing.~Then 
\be
\label{main_lower}
M=\Omega(n^d).
\ee
\end{theorem}

Before stating our learning algorithm and sample complexity, we need the following lemmas. 
\begin{lemma}
\label{lemma:characterization_f_gps}
Choose any $f\in \calP_q(n,d)$ and $c\in \ZZ_q$. Then either $f(x)$ is a constant function or the fraction of inputs $x\in \{0,1\}^n$ such that $f(x)\ne c$  is at least $1/2^d$.
\end{lemma}

\begin{proof}
We shall use the following simple fact.
\begin{prop}
\label{prop:constant}
Consider a function $f\, : \, \{0,1\}^n\to \ZZ_q$ specified as a polynomial 
\be
f(x) = \sum_{J\subseteq [n] } \alpha_J \prod_{j\in J} x_j {\pmod {q}}.
\ee
Here $\alpha_J\in \ZZ_q$ are coefficients.
The function $f$ is constant if and only if  $\alpha_J =0 {\pmod q}$ for all non-empty subsets $J\subseteq [n]$.
\end{prop}
\begin{proof}
If $\alpha_J =0 {\pmod q}$ for all non-empty subsets $J$ then $f(x)=f(0^n) {\pmod q}$ for all $x$, that is,
$f$ is constant. Conversely,
suppose $f$ is constant.
Choose a subset $J\subseteq [n]$. We can consider $J$ as an $n$-bit string
with the Hamming weight $|J|$ such that $J_i=1$
if $i\in J$ and $J_i=0$ otherwise. 
If $|J|=1$ then $f(J)=f(0^n) + \alpha_J {\pmod q}$
and thus $\alpha_J=0 {\pmod q}$ for all subsets $J$
with $|J|=1$. Suppose we have already proved that $\alpha_J=0 {\pmod q}$
for any subset with $1\le |J|\le w$. If $|J|=w+1$ then 
$f(J)=f(0^n)+ \alpha_J {\pmod q}$ and thus $\alpha_J=0 {\pmod q}$ for all subsets $J$
with $|J|=w+1$. Proceeding inductively proves the claim.
\end{proof}

We shall prove Lemma~\ref{lemma:characterization_f_gps}  by induction in $n$. The base case of induction is $n=d$. 
Clearly, a non-constant function $f\, : \, \{0,1\}^d \to \ZZ_q$ takes a value different from $c$ at least one time, that is,
the fraction of inputs $x\in \{0,1\}^d$ such that $f(x)\ne c$ is at least $1/2^d$.

Suppose $n>d$ and $f\in \calP_q(n,d)$ is not a constant function.
Let $d'$ be the maximum degree of non-zero monomials in $f$.
Clearly $1\le d'\le d$.
Suppose $f$ contains a monomial
$\alpha_S \prod_{j\in S}x_j$ where $\alpha_S \in \ZZ_{q}\setminus \{0\}$
and $|S|=d'$.
Since $|S|<n$, one can choose a variable $x_i$ 
with
$i\in [n]\setminus S$. Let $g_a\, : \{0,1\}^{n-1}\to \ZZ_{q}$
be a function obtained from $f$ by setting the variable $x_i$ to a constant value $a\in \{0,1\}$.
Clearly, $g_a\in \calP_q(n-1,d)$.
The coefficients of the monomial $\prod_{j\in S}x_j$ in $g_0$ and $g_1$
 are $\alpha_S$ and $\alpha_S + \alpha_{S\cup \{i\}} {\pmod q}$ respectively.
 However, $\alpha_{S\cup \{i\}} =0 {\pmod q}$ since otherwise $f$ would contain
 a monomial $x_i \prod_{j\in S} x_j$ of degree larger than $d'$. 
We conclude that both $g_0$ and $g_1$
contain a non-zero monomial $\alpha_S \prod_{j\in S}x_j$.
 By Proposition~\ref{prop:constant},
$g_0$ and $g_1$ are not constant functions. 
Since $g_0$ and $g_1$ are degree-$d$ polynomials in $n-1$ variables, 
the induction hypothesis gives
\be
\Pr_y[g_a(y)\ne c] \ge  \frac1{2^{d}}.
\ee
Here $y\in \{0,1\}^{n-1}$ is picked uniformly at random. 
Thus 
\be
\Pr_x[f(x)\ne c] = \frac12\left[ \Pr_y[g_0(y)\ne c]+ \Pr_y[g_1(y)\ne c]\right] \ge  \frac1{2^d}.
\ee
Here $x\in \{0,1\}^n$ is picked uniformly at random. This proves the induction step.
\end{proof}

With this lemma, we are now ready to prove Theorem~\ref{thm:generalizedphasedegreed}. In the section below we first describe our learning algorithm and in the next section we prove the theorem by proving the sample complexity upper bound.

\subsubsection{Learning Algorithm in  Theorem~\ref{thm:generalizedphasedegreed}}
We are now ready to state our learning algorithm. As in Section~\ref{sec:learning_bps_separable_measurements} for learning binary phase states with separable measurements, we learn generalized phase states through examples containing information about the derivatives of $f(x)$. The crucial difference between the binary phase state learning algorithm and the generalized setting is, in the binary case, we obtained a measurement outcome $b_y$ that corresponded to $b_y=f(0y)-f(1y)$, however in the generalized scenario, we obtain a measurement outcome $b'_y$ that satisfies $f(0y)-f(1y)\neq b'_y$. Nevertheless, we are able to still learn $f$ using such measurement outcomes which we describe in the rest of the section.  

We now describe the learning algorithm. We carry out the algorithm in $n$ rounds, which we index by $t$. For simplicity, we describe the procedure for the first round. Suppose we measure qubits $2,3,\ldots,n$ of the state $\ket{\psi_f}$ in the $Z$-basis. Let $y\in \{0,1\}^{n-1}$ be the measured bit string. Note that the probability distribution of $y$ is uniform.  The post-measurement state of qubit $1$ is 
\be
    \ket{\psi_{f,y}} =\frac1{\sqrt{2}} ( \omega_q^{f(0y)} \ket{0} + \omega_q^{f(1y)} \ket{1})
    \label{eq:post_measurement_state_gps_sep_meas}
\ee
For each $b\in \ZZ_q$ define a single-qubit state 
\begin{equation}
    \ket{\phi_b} = \frac{1}{\sqrt{2}}(\ket{0} -\omega_q^{b} \ket{1} )   
\label{phi_b}
\end{equation}
Using the identity 
$\sum_{b\in \ZZ_q} \omega_q^b = 0$ one gets 
\be
    I = \frac2{q} \sum_{b\in \ZZ_q} \ket{\phi_b} \bra{\phi_b}
    \label{eq:POVM_sep_meas_gps}
\ee
One can view Eq.~(\ref{eq:POVM_sep_meas_gps}) as a single-qubit POVM with $q$ elements $(2/q)|\phi_b\ra\la \phi_b|$. Let $\calM$ be the single-qubit measurement described by this POVM. Applying $\calM$ to the state $\ket{\psi_{f,y}}$  returns an outcome $b\in \ZZ_q$ with the probability
\be
    \Pr(b|y):= \frac2{q} |\la \phi_b|\psi_{f,y}\ra|^2 = \frac1{2q}\left| 1 - \omega_q^{f(1y)-f(0y)-b}\right|^2.
    \label{eq:prob_sample_sep_meas_gps}
\ee
Clearly, $\Pr(b|y)$ is a normalized probability distribution, $\sum_{b\in \ZZ_q} \Pr(b|y)=1$. Furthermore, 
\begin{align}
    f(1y)-f(0y)=b  \quad \text{implies } & \Pr(b|y)=0,\\
    f(1y)-f(0y)\ne b  \quad \text{implies } &  \Pr(b|y)\ge 
    \frac2{q} \sin^2{(\pi/q)} = \Omega(1/q^3).
    \label{eq:inference_sample_sep_meas_gps}
\end{align}
To conclude, the combined $n$-qubit measurement consumes one copy of the state $\ket{\psi_f}$ and returns a pair $(y,b) \in \{0,1\}^{n-1} \times \ZZ_q$ such that 
\begin{equation}
    f(1y)-f(0y) \ne b   
\label{eq:constraint_sample_sep_meas_gps}    
\end{equation}
with certainty and all outcomes $b$ satisfying Eq.~(\ref{eq:constraint_sample_sep_meas_gps}) appear with a non-negligible probability. 
Define a function $g\, : \, \{0,1\}^{n-1}\to \ZZ_q$
such that
\be
    \label{g(z)}
    g(y) = f(1y)-f(0y).
\ee
We claim that $g$ is a degree-$(d-1)$ polynomial,
that is, $g\in \calP_q(n-1,d-1)$.
Indeed, it is clear that $g(y)$ is a degree-$d$ polynomial. Moreover, 
all degree-$d$ monomials in $f(x)$ that do not contain the variable $x_1$ appear in $f(1y)$ and $f(0y)$
with the same coefficient. Such monomials do not contribute to $g(y)$. A degree-$d$ monomial in $f(x)$ that contains the variable $x_1$ contributes a degree-$(d-1)$ monomial to $g(y)$.
Thus $g \in \calP_q(n-1,d-1)$, as claimed. 

From Eq.~(\ref{eq:constraint_sample_sep_meas_gps}) one infers a constraint
\be
    g(y)\ne b
    \label{eq:constraint2_sample_sep_meas_gps}
\ee
whenever the combined $n$-qubit measurement of $\ket{\psi_f}$ returns an outcome $(y,b)$. Suppose we repeat the above process $m$ times  obtaining constraints 
\be
    g(y^{(k)})\ne b^{(k)}, \qquad k=1,2,\ldots,m.
    \label{eq:system_constraints_sep_meas_gps}
\ee
This consumes $m$ copies of $\ket{\psi_f}$. We claim that the probability of having more than one polynomial 
$g\in \calP_q(n-1,d-1)$ satisfying the constraints Eq.~(\ref{eq:system_constraints_sep_meas_gps}) is exponentially small if we choose
\be
    m=O(q^3 \log{(q)} 2^d n^{d-1}).
    \label{eq:sample_complexity_sep_meas_gps}
\ee

\subsubsection{Sample Complexity bound in Theorem~\ref{thm:generalizedphasedegreed}}
Define a probability distribution  $\pi(\vec{y},\vec{b})$ where 
\begin{equation}
    \vec{z}=(y^{(1)},\ldots,y^{(m)})\in \{0,1\}^{(n-1) m} \quad \mbox{and} \quad
\quad \vec{b}=(b^{(1)},\ldots,b^{(m)}) \in  (\ZZ_q)^{\times m}    
\end{equation}
such that $y^{(j)}$ are picked uniformly at random and $b^{(k)}$ are sampled from the distribution $\Pr(b^{(k)}|y^{(k)})$ defined in Eq.~(\ref{eq:prob_sample_sep_meas_gps}). For each polynomial $h\in \calP_q(n-1,d-1)$ define an event
\begin{equation}
    \mathsf{BAD}(h)=\{ (\vec{y},\vec{b})\, : \, h(y^{(k)}) \ne b^{(k)}  \quad \mbox{for all} \quad k \in [m]\}.    
\end{equation}
We claim that 
\be
    \Pr[\mathsf{BAD}(h)]:= \sum_{(\vec{y},\vec{b})\in \mathsf{BAD}(h)}\; \pi(\vec{y},\vec{b}) \le \left[1-\Omega(2^{-d} q^{-3})\right]^m
    \label{eq:prob_bad_event_sep_meas_gps}
\ee
for any $h\ne g$. Indeed, consider some fixed $k \in [m]$. The event $b^{(k)} \ne h(y^{(k)})$ occurs automatically if $h(y^{(k)})= g(y^{(k)})$. Otherwise, if $h(y^{(k)}) \ne g(y^{(k)})$, the event $b^{(k)} \ne h(y^{(k)})$ occurs with the probability at most $1-\Omega(1/q^3)$ since $b^{(k)} = h(y^{(k)})$ with the probability at least $\Omega(1/q^3)$ due to Eq.~(\ref{eq:inference_sample_sep_meas_gps}). It follows that 
\begin{align}
\Pr_{y^{(k)},b^{(k)}}[h(y^{(k)})\ne b^{(k)}] & \le \Pr_{y^{(k)}}[h(y^{(k)})= g(y^{(k)})] + \Pr_{y^{(k)}}[h(y^{(k)}) \ne g(y^{(k)})] \left( 1- \Omega(1/q^3) \right) \\
& = 1 - \Pr_{y^{(k)}}[h(y^{(k)})\ne g(y^{(k)})] \cdot \Omega(1/q^3).
\end{align}
If $h$ and $g$ are equivalent then $h(y)=g(y)+c$ for some constant $c\in \ZZ_q$. Note that $c\ne 0$ since we assumed $h\ne g$. In this case 
\begin{equation}
    \Pr_{y^{(k)}}[h(y^{(k)}) \ne g(y^{(k)})]=1.    
\end{equation}
If $h$ and $g$ are non-equivalent,  apply Lemma~\ref{lemma:characterization_f_gps} to a non-constant degree-$(d-1)$ polynomial $h-g$. It gives
\begin{equation}
    \Pr_{y^{(k)}}[h(y^{(k)})\ne g(y^{(k)})]\ge \frac1{2^{d-1}}.    
\end{equation}
In both cases we get 
\begin{equation}
    \Pr_{y^{(k)},b^{(k)}}[h(y^{(k)})\ne b^{(k)}]\le  1- \Omega(2^{-d} q^{-3}),    
\end{equation}
which proves Eq.~(\ref{eq:prob_bad_event_sep_meas_gps}) since the pairs $(y^{(k)},b^{(k)})$ are i.i.d. random variables.

As noted earlier in the preliminaries, observe that $|\calP_q(n-1,d-1)|\le q^{O(n^{d-1})}
= 2^{O(\log{(q)} n^{d-1})}$. By the union bound, one can choose $m=O(2^{d} q^3 \log{(q)}n^{d-1})$ such that 
\be
    \Pr\left[\bigcup_{h\in \calP_q(n-1,d-1)\setminus g}\; \mathsf{BAD}(h)\right]\le 2^{O( \log{(q)}n^{d-1})}  \left[1-\Omega(2^{-d}q^{-3})\right]^m \le 2^{-\Omega(n)}.
    \label{eq:prob_bad_event_polynomials_sep_meas_gps}
\ee
In other words, the probability that $g$ is the unique element of $\calP_q(n-1,d-1)$ satisfying all the constraints Eq.~(\ref{eq:system_constraints_sep_meas_gps}) is at least $1-2^{-\Omega(n)}$. One can identify such polynomial $g$ by  checking the constraints Eq.~(\ref{eq:system_constraints_sep_meas_gps}) for every $g\in \calP_q(n-1,d-1)$. If the constraints are satisfied for more than one polynomial, declare a failure.

At this point we have learned a polynomial $g\in \calP_q(n-1,d-1)$ such that $f(1y)-f(0y)=g(y)$ for all $y\in \{0,1\}^{n-1}$. For simplicity,  we ignore the exponentially small  failure probability. 
Applying the same protocol $n$ times to copies of the quantum state $\ket{\psi_f}$ by a cyclic shift of qubits, one can learn polynomials $g_0,g_1,\ldots,g_{n-1} \in \calP_q(n-1,d-1)$ such that 
\begin{equation}
    f(C^i(1y))-f(C^i(0y)) = g_i(y) \quad \mbox{for all} \quad i\in [n] \quad \mbox{and} \quad y \in \{0,1\}^{n-1},
    \label{eq:cyclic_shift}    
\end{equation}
where $C$ is the cyclic shift of $n$ bits.  This consumes $M=O(nm)= O(2^d q^3\log{(q)} n^d)$ copies of the state $\ket{\psi_f}$. We can assume wlog that $f(0^n)=0$ since our goal is to learn $f(x)$ modulo a constant shift. Suppose we have already learned values of $f(x)$ for all bit strings $x$ with the Hamming weight $|x|\le w$ (initially $w=0$).  Any bit string $x$ with $|x|=w+1$ can be represented as $x=C^i(1y)$ for some $y\in \{0,1\}^{n-1}$ such that $|y|=w$. Now Eq.~(\ref{eq:cyclic_shift}) determines $f(x)$ since $|C^i(0y)|=|y|=w$ so that $f(C^i(0y))$ is already known and the polynomial $g_i(y)$ has been learned. Proceeding inductively one can learn $f(x)$ for all $x$. 

It remains to note that the POVM Eq.~(\ref{eq:POVM_sep_meas_gps}) is a probabilistic mixture of projective single-qubit measurements whenever $q$ is even. Indeed, in this case the 
states $|\phi_b\ra$ and $|\phi_{b+q/2}\ra=Z|\phi_b\ra$ form an orthonormal basis of a qubit, see Eq.~(\ref{phi_b}). Thus the POVM defined in Eq.~(\ref{eq:POVM_sep_meas_gps}) can be implemented by picking a random uniform $b\in \ZZ_q$ and measuring a qubit in the basis $\{ |\phi_b\ra, Z|\phi_b\ra\}$. Thus the learning protocol only requires single-qubit unitary gates and measurements in the standard basis. 

The lower bound in the proof of Theorem~\ref{thm:generalizedphasedegreed} follows in a straightforward manner from the lower bound  for binary phase states.
Indeed, suppose 
\[
f'(x)=\sum_{J\in [n]} \alpha_J \prod_{j\in J} x_j {\pmod 2}
\]
is an $\FF_2$-valued degree-$d$ polynomial, $f'\in \calP(n,d)$.
Suppose 
$q=2r$ for some integer $r$. 
Define a polynomial
\[
f(x) = r f'(x) {\pmod q}.
\]
Clearly $f\in \calP_q(n,d)$
and
$\omega_q^{f(x)} = (-1)^{f'(x)}$ for all $x$, that is the binary phase state corresponding to $f'$ coincides with the generalized phase state corresponding to $f$. Using Theorem~\ref{thm:lowerboundbinaryphase}, we obtain a lower bound of $M=\log |\calP(n,d)|=
\Omega(n^d)$ for learning $\psi_f$. This concludes the proof of Theorem~\ref{thm:generalizedphasedegreed}.

\subsection{Learning stabilizer states}
We now describe how the algorithm stated in Theorem~\ref{thm:generalizedphasedegreed} could be used to learn any $n$-qubit stabilizer state (produced by a Clifford circuit applied to $\ket{0^n}$ state) using separable measurements. Note that we can learn a subclass of stabilizer states called graph states (which are simply binary phase states with $d=2$)  using Algorithm~\ref{algo:learn_bps_separable_measurements} with the sample complexity of $O(n^2)$ (as shown in Theorem~\ref{thm:binaryphasedegreed}).

From a result in \cite{dehaene2003clifford}, we know that a stabilizer state can be represented as follows
\begin{equation}
    \ket{\psi} = \frac{1}{\sqrt{|A|}} \sum \limits_{x \in A} i^{\ell(x)} (-1)^{q(x)} \ket{x},
\end{equation}
where $A$ is an affine subspace of $\FF_2^n$, $\ell: \FF_2^n \rightarrow \FF_2$ is a linear function and $q: \FF_2^n \rightarrow \FF_2$ is quadratic function. Clearly, an alternate form is a generalized phase state with degree-$2$
\begin{equation}
    \ket{\psi_f} = \frac{1}{\sqrt{|A|}} \sum \limits_{x \in A} i^{f(x)} \ket{x}
\end{equation}
where the summation is over $A$ instead of the entire $\FF_2^n$, and the function $f: \FF_2^n \rightarrow \ZZ_4$ has its coefficients corresponding to the quadratic monomials take values in $\{0,2\}$. We can now learn this using separable measurements as stated in the following statement as opposed to entangled measurements as required by Bell sampling \cite{montanaro2017learning}.

\begin{corollary}
\label{thm:learn_stab_states_sep_meas}
There exists an algorithm that uses $M=O(n^{2})$ copies of a stabilizer state $\ket{\psi_f} = \frac{1}{\sqrt{|A|}} \sum \limits_{x \in A} i^{f(x)} \ket{x}$ with an unknown polynomial $f\in \calP_4(n,2)$  and  outputs a polynomial $g\in \calP_4(n,2)$ such that  $g$ is equivalent to $f$ with the probability at least $1-2^{-\Omega(n)}$. The quantum part of the algorithm requires only single-qubit unitary gates and measurements in the standard basis.
\end{corollary}
\begin{proof}
The subspace $A$ of an unknown stabilizer state can be denoted as $a + S_A$ where $a \in \FF_2^n$ is a translation vector and $S_A$ is a linear subspace of $\FF_2^n$. To learn $a$ and a basis of the subspace $S_A$, it is enough to measure $O(n \log n)$ copies of $\ket{\psi_f}$ in the computational basis. This in turn defines a subset of the $n$ directions $\{ e_i \}$ along which we need to search for non-zero monomials in the partial derivatives of $f$. We can now use the learning algorithm in Theorem~\ref{thm:generalizedphasedegreed} to learn the unknown stabilizer state using $O(n^2)$ copies with the desired probability.
\end{proof}

\section{Learning noisy phase states}
\label{sec:learningnoisy}
In this section, we consider learning algorithms in the presence of noise (in particular we consider global depolarizing noise, local depolarizing noise and local depolarizing noise when the phase state has additional graph structure). 

\subsection{Global depolarizing noise}\label{sec:globalnoise}
Let 
$$
\ket{\psi_f}=\frac{1}{2^{n/2}}\sum_{x\in \01^n}(-1)^{f(x)}\ket{x}
$$
where $f:\01^n\rightarrow \01$ is a degree-$2$ polynomial in $\FF_2$. For simplicity, we assume $f(x)=x^\top A x$ (where $A\in \FF_2^{n\times n}$ is upper triangular). Suppose we are given noisy copies of $\ket{\psi}$ of the form
$$
{\psi_f}=(1-\varepsilon)\cdot \ketbra{\psi_f}{\psi_f}+\varepsilon \cdot \mathbb{I}/2^n
$$
for some $\varepsilon>0$, then how many copies of  ${\psi_f}$ are necessary and sufficient to learn $f$?

Below, we observe the following theorem:
\begin{theorem}
\label{lem:noisystablearning}
Let $\varepsilon>0$ be a constant. Given $2n^{1+\delta}$ copies of $\psi_f$ (with error $\varepsilon$)  and $n\cdot 2^{O(n/\log \log n)}$ time (for some constant $\delta \in (0,1)$ dependent on $\varepsilon$), there exists a procedure to learn $A$.
\end{theorem} 
Our argument crucially uses the result of Lyubashevsky.
\begin{theorem}[\cite{lyubashevsky2005parity}]
\label{thm:lyubashevsky}
We are given $n^{1+\delta}$ ordered pairs $(a_i,\ell_i)$ where $a_i$ are chosen uniformly and independently at random from the set $\01^n$ and for some $c\in \01^n$, 
\[ \ell_i=\begin{cases} 
      c\cdot a_i \pmod 2 & \text{w.p. } 1/2+\eta \\
      1+c\cdot a_i \pmod 2 & \text{w.p. } 1/2-\eta 
   \end{cases}
\]
If $\eta> 2^{-(\log n)^\delta}$ for constant $\delta <1$, then there is an algorithm that can recover $c$ in time $2^{O(n/\log \log n)}$ with high probability. 
\end{theorem}
We also use the following simple lemma.
The procedure above is an application of Bell sampling~\cite{montanaro2017learning} to the pure state $\ket{\psi_f}^{\otimes 2}$ and mixed state $\psi_f^{\otimes 2}$.  We now prove our main theorem statement. 
\begin{proof}[Proof of Theorem~\ref{lem:noisystablearning}]
 For simplicity let $B=A+A^\top$. One way to view Lemma~\ref{lem:bellsampling} is that, it uses two copies of $\psi_n$ and produces a $(z,w_z)\in \01^{2n}$ such that $(w_z)_i=B^i\cdot z$ (where $B^i$ is the $i$ row of $B$) with probability $(1-\varepsilon)^2$ and is a uniformly random bit $b\in \01$ with probability $1-(1-\varepsilon)^2$. In particular,
 \[ (w_z)_i=\begin{cases} 
      B^i\cdot z & \text{w.p. }  1/2+(1-\varepsilon)^2/2\\
      1+ B^i\cdot z & \text{w.p. } 1/2-(1-\varepsilon)^2/2.
   \end{cases}
\]
Hence two copies of $\psi_f$ can be used to obtain $(z,(w_z)_1),\ldots,(z,(w_z)_n)$.   
So the learning algorithm first uses $T=n^{1+\delta}$ many copies of $\psi_f$ and produces
\begin{align*}
    &(z^1,(w_{z^1})_1),\ldots,(z^1,(w_{z^1})_n)\\
    &(z^2,(w_{z^2})_1),\ldots,(z^2,(w_{z^2})_n)\\
    &\hspace{10mm}\vdots\\
    &(z^T,(w_{z^T})_1),\ldots,(z^T,(w_{z^T})_n).
\end{align*}
Each column above (i.e., $(z^1,(w_{z^1})_1),\ldots,(z^T,(w_{z^T})_1)$) can be now be given as input to the algorithm of Theorem~\ref{thm:lyubashevsky} where $\eta=(1-\varepsilon)^2/2$ is a constant (and $\delta>0$ is also a tiny constant), which  produces $B^i$ with high probability.\footnote{The high probability in Theorem~\ref{thm:lyubashevsky} is in fact inverse exponential in $n$.} Hence feeding all the $n$ different columns to Theorem~\ref{thm:lyubashevsky} allows the algorithm to learn $B^1,\ldots,B^n$ explicitly. The overall sample complexity is $2n^{1+\delta}$ and time complexity is $n\cdot 2^{O(n/\log\log n)}$. Once we learn the off-diagonal elements of $A$ (since above we only obtain information of $A+A^\top$ which zeroes the diagonal entries of $A$), a learning algorithm, applies the operation $\ket{x}\rightarrow (-1)^{x_{ij}}\ket{x}$ if $A_{ij}=1$ for $i\neq j$. Repeating this for all the $n(n-1)/2$ different $i\neq j$, the resulting quantum state is $\sum_x (-1)^{\sum_i x_i A_{ii}}\ket{x}$ which is a linear phase state, and we can learn using Bernstein-Vazirani algorithm. 
\end{proof}

\subsection{Local Depolarizing Noise}

Let us now show that learning phase states subject to a local depolarizing noise
has an exponential sampling complexity in the worst case. 

\begin{theorem}
\label{thm:localdepolarizing}
For every $\varepsilon>0$, learning degree-$2$ phase states with $\varepsilon$-local depolarizing noise has sample complexity $\Omega((1-\varepsilon)^n)$.
\end{theorem}
\begin{proof}
 Let $\calD_1$ be a single-qubit depolarizing channel 
that implements the identity  with probability $1-\varepsilon$
and outputs a maximally mixed state with probability $\varepsilon$,
\[
\calD_1(\rho)=(1-\varepsilon) \rho + \varepsilon \mathrm{Tr}(\rho) \frac{I}2.
\]
Let $\calD=\calD_1^{\otimes n}$ be the $n$-qubit depolarizing channel. 
Consider $n$-qubit GHZ-like states 
\[
|\phi^\pm\rangle = (|0^n\rangle \pm |1^n\rangle)/\sqrt{2}.
\]
Using the identity $\calD_1(|0\rangle\langle 1|)=(1-\varepsilon)|0\rangle\langle 1|$ one gets
\[
\calD(|\phi^+\rangle\langle \phi^+|)-\calD(|\phi^-\rangle\langle \phi^-|)=(1-\varepsilon)^n (
|\phi^+\rangle\langle \phi^+| - |\phi^-\rangle\langle \phi^-|)
\]
which implies
\begin{equation}
\label{GHZdistance}
\| \, \calD(|\phi^+\rangle\langle \phi^+|)-\calD(|\phi^-\rangle\langle \phi^-|) \, \|_1
\le 2(1-\varepsilon)^n.
\end{equation}
It follows that the trace distance between $k$ copies of the states
$\calD(|\phi^+\rangle\langle \phi^+|)$ and $\calD(|\phi^-\rangle\langle \phi^-|)$
is at most $2k(1-\varepsilon)^n$. By Helstrom theorem, these states cannot be
distinguished reliably unless $k=\Omega((1-\varepsilon)^{-n})$.
Next we  observe that $|\phi^\pm\rangle$ are degree-two phase states
modulo single-qubit rotations $R_x=e^{i(\pi/4)X}$.
Indeed, suppose $n=1{\pmod 4}$.
Then
a simple algebra shows that
\begin{equation}
\label{Rx}
R_x^{\otimes n}|\phi^+\rangle=e^{i\pi/4} |\psi_f\rangle
\quad \mbox{and} \quad  
R_x^{\otimes n}|\phi^-\rangle=e^{-i\pi/4}|\psi_g\rangle,
\end{equation}
where $\psi_f$ and $\psi_g$ are $n$-qubit phase states
associated with  degree-two  polynomials 
\[
f(x)=\sum_{1\le i<j\le n} x_i x_j {\pmod 2}
\quad \mbox{and} \quad g(x)=f(x)+\sum_{i=1}^n x_i {\pmod 2}.
\]
Since the depolarazing channel $\calD$ commutes with single-qubit unitary operators, Eqs.~(\ref{GHZdistance},\ref{Rx})
give
\[
\| \, \calD(|\psi_f\rangle\langle \psi_f|)-\calD(|\psi_g\rangle\langle \psi_g|) \, \|_1
\le 2(1-\varepsilon)^n.
\]
Thus
$k$ copies of the noisy phase states $\calD(|\psi_f\rangle\langle \psi_f|)$
and $\calD(|\psi_g\rangle\langle \psi_g|)$ cannot be
distinguished reliably unless $k=\Omega((1-\varepsilon)^{-n})$.
We conclude that the sampling complexity of learning
phase states subject to local $\varepsilon$-depolarizing noise
is at least $\Omega((1-\varepsilon)^{-n})$, which is 
exponentially large in $n$ for any constant error
rate $\varepsilon>0$.
\end{proof}

\subsection{Local depolarizing noise and small graph degree}
Although learning phase states with local depolarizing is hard in general, a restricted class of states is again easy to learn using the same technique from Section~\ref{sec:globalnoise}. Recall from that section the notation $f(x)=x^\top Ax$ and $B=A+A^\top$, where $A$ is upper triangular and $B$ is symmetric. Interpreting $B$ as the adjacency matrix of a graph, we define the graph-degree of $f$ to be $\text{gd}(f)=\max_i|B^i|$, where $B^i$ is the $i^{\text{th}}$ row of $B$. The graph-degree of $f$ is also one less than the maximum stabilizer weight of the stabilizer state $\ket{\psi_f}$.

It is possible to learn phase states suffering from local depolarizing using only a few copies if their graph degree is promised to be small.

\begin{theorem}
Let $\varepsilon>0$ be a constant, $\mathcal{D}$ be local depolarizing noise on $n$ qubits with strength $\varepsilon$, and $f(x)$ be a degree-2 polynomial with $\text{gd}(f)<(\log n)^{\delta'}$ for some constant $\delta'$. Given $2n^{1+\delta}$ copies of $\mathcal{D}(\psi_f)$ and $n2^{O(n/\log \log n)}$ time (for some constant $\delta \in (0,1)$ dependent on $\varepsilon$ and $\delta'$), there exists a procedure to learn $A$.
\end{theorem}
\begin{proof}
Suppose we apply Bell-sampling, Lemma~\ref{lem:bellsampling}, to phase states suffering from local depolarizing noise $\mathcal{D}$. Recall that this involves measuring the two-body operators $Z\otimes Z$ and $X\otimes X$ on corresponding pairs of qubits from two-copies of the state. Since these are two-qubit operators, they are randomized by the noise with probability $1-(1-\varepsilon)^2$. 

One use of Bell-sampling on two copies of the state gives $(z,w_z)\in\{0,1\}^{2n}$, where each of the $2n$ bits is correct (i.e.~is the same as we would get without noise) with probability $(1-\varepsilon)^2$ and uniformly random with probability $1-(1-\varepsilon)^2$. Therefore, $(w_z)_i+B^i\cdot z$, which is a sum of $|B^i|+1$ bits, is $0$ (mod 2) with probability $(1-\varepsilon)^{2(|B^i|+1)}$ and uniformly random otherwise. Just as in the proof of Theorem~\ref{lem:noisystablearning}, we can apply Theorem~\ref{thm:lyubashevsky}, now with $\eta=(1-\varepsilon)^{2(\text{gd}(f)+1)}>(1-\varepsilon)^{2((\log n)^{\delta'}+1)}$, to learn $B$. Once we learn $B$, we can learn the diagonal elements of $A$, using the same procedure as in the proof of Theorem~\ref{lem:noisystablearning}.
\end{proof}

\section{Applications} \label{sec:applications}
In this section, we describe how the algorithms for learning phase states (see Table~\ref{tab:summary_results_paper}) can be used to learn the quantum circuits that produce binary phase states in Section~\ref{sec:learning_circuits_bps} and generalized phase states in Section~\ref{sec:learning_circuits_gps}. For each, we firstly describe the quantum circuits that produce the phase states of interest followed by our results for learning these quantum circuits.

\subsection{Learning quantum circuits producing binary phase states}
\label{sec:learning_circuits_bps}
We consider a $n$-qubit circuit  $\mathsf{C}$ produced from the set of gates $S=\{\mathsf{H},\mathsf{Z},\mathsf{CZ},\mathsf{CCZ},\ldots,\mathsf{C^{d-1}Z}\}$ where $\mathsf{H}$ is the Hadamard gate and $\mathsf{C^{d-1}Z}$ denotes the controlled-Z gate with $(d-1)$ control qubits. We will actually restrict ourselves to circuits $\mathsf{C}$ which start and end with a column of Hadamard gates over all $n$-qubits, with its internal part $\mathsf{C'}$ containing gates from $S \setminus \mathsf{H}$. We then have the following statement regarding the states produced by $\mathsf{C}$
\begin{prop}\label{prop:circuits_bps}
Let $\mathsf{C}$ be an $n$-qubit quantum circuit, starting and ending with a column of Hadamard gates, with its internal part $\mathsf{C'}$ only containing $s$ gates from $\{\mathsf{Z},\mathsf{CZ},\mathsf{CCZ},\ldots,\mathsf{C^{d-1}Z}\}$,~then
\begin{equation}
    \ket{\psi_f} = \mathsf{C'}\ket{+}^{\otimes n} = \frac{1}{\sqrt{2^n}} \sum \limits_{x \in \Fset{n}} (-1)^{f(x)} \ket{x},
\end{equation}
where the corresponding Boolean function $f \in \calP(n,d,s)$.
\end{prop}
\begin{proof}
We follow a proof strategy similar to that in \cite[Prop.~1]{montanaro2007distinguishability} and \cite[Appendix~B]{bremner2016average}, which treated the case of $d=3$. Let $\mathsf{Z}_{i_1}$ be the $\mathsf{Z}$ gate acting on the $i_1$th qubit, $\mathsf{CZ}_{i_1,i_2}$ be the controlled-Z gate with $i_1$th qubit as the control, and similarly $\mathsf{C^{d-1}Z}_{i_1,i_2,\dots,i_d}$ with controls on $(i_1,\ldots,i_{d-1})$ qubits. We note that for any $x \in \FF_2^n$, 
\begin{align}
    \braket{x}{\mathsf{Z}_{i_1}}{x} = (-1)^{x_{i_1}}, \, \braket{x}{\mathsf{CZ}_{i_1,i_2}}{x} = (-1)^{x_{i_1}x_{i_2}}, \, \braket{x}{\mathsf{C^{d-1}Z}_{i_1,i_2,\ldots,i_d}}{x} = (-1)^{x_{i_1}x_{i_2}\ldots x_{i_d}}.
\end{align}
As all these gates are diagonal, we can obtain an expression for $\braket{x}{\mathsf{C'}}{x}$ by simply multiplying the expressions of $\braket{x}{G}{x}$ for the different gates $G$ in $\mathsf{C'}$. To complete the proof, we note that the $\FF_2$-degree of $f$ is $k$ if and only if $\mathsf{C^{k-1}Z}$ is the gate with highest controls present in $\mathsf{C'}$ and the number of terms in $f$ is at most the number of gates applied in $\mathsf{C'}$.
\end{proof}
Note that the states produced in Proposition~\ref{prop:circuits_bps} are exactly the binary phase states corresponding to Boolean functions $f$ given by~Eq.~\eqref{eq:anf_boolean_function}. Some special classes of circuits included in the above statement are Clifford circuits which produce graph states for $d=2$, and $\mathsf{IQP}$ circuits for $d=3$. We observe from the above proposition that there can be more than one quantum circuit $\mathsf{C}$ corresponding to a given polynomial $f \in \calP(n,d,s)$. As the internal gates of $\mathsf{C'}$ in $S$ commute, these gates can be reordered arbitrarily while still producing the same Boolean function $f$.

To learn a circuit representation of $\mathsf{C}$ from samples, we have the following result.

\begin{theorem}
Let $\mathsf{C}$ be an unknown $n$-qubit quantum circuit, starting and ending with a column of Hadamard gates, with its internal part $\mathsf{C'}$ only containing gates from $\{\mathsf{Z},\mathsf{CZ},\mathsf{CCZ},\ldots,\mathsf{C^{d-1}Z}\}$. A circuit representation of $\mathsf{C}$ can then be learned through $O(n^d)$ queries to $\mathsf{C}$ and using only separable measurements. This can be improved to $O(n^{d-1})$ queries to $\mathsf{C}$ and using entangled measurements.
\end{theorem}
\begin{proof}
From Proposition~\ref{prop:circuits_bps}, we note the correspondence between $\mathsf{C'}$ and the binary phase state $\ket{\psi_f}$. From Theorem~\ref{thm:binaryphasedegreed}, we have that we can learn the $\FF_2$ representation of $f$ corresponding to such a state, using $O(n^d)$ separable measurements. Given $O(n^d)$ uses of the unknown $\mathsf{C}$, we thus learn $f$ from samples generated by applying $\mathsf{H}^{\otimes n} \mathsf{C}$ on $\ket{0}^{\otimes n}$ followed by separable measurements. We obtain a circuit representation of $\mathsf{C'}$ (and hence $\mathsf{C}$ which is $\mathsf{H}^{\otimes n} \mathsf{C'}\mathsf{H}^{\otimes n}$) by inserting gates $\mathsf{C^{|J|-1}Z}_{i_1,i_2,\ldots,i_{|J|}}$ (where $i_1,i_2,\ldots,i_{|J|} \in J$) for each monomial $\prod_{i \in J \subseteq [n]} x_i$, characterized by set $J$, present in $f$. The result for entangled measurements is obtained through application of Theorem~\ref{thm:entangledupperbound}.
\end{proof}

\subsection{Learning circuits containing diagonal gates in the Clifford hierarchy} \label{sec:learning_circuits_gps}
For any two $n$-qubit unitaries $U,V\in\mathcal{U}(n)$, let $[U,V]=UVU^\dag V^\dag$ denote the group commutator, and let $P(n)=\langle iI,X_j,Y_j,Z_j:j\in[n]\rangle$ be the $n$-qubit Pauli group. The $d^{\text{th}}$ level of the Clifford hierarchy on $n$-qubits, denoted $C_d(n)$, is defined inductively
\begin{align}
C_1(n) &= P(n),\\
C_d(n) &= \{U\in\mathcal{U}(n):[U,p]\in C_{d-1}(n),\forall p\in P(n)\}.
\end{align}
The second level $C_2(n)$ is the $n$-qubit Clifford group, while higher levels are not groups at all, as they fail to be closed. In general, $C_d(n)$ includes $C_{d-1}(n)$ and more -- for instance, some gates that are in the set $C_{d}(n)$ and not in any lower level of the hierarchy are $Z^{1/2^{d-1}}$ and the controlled-$Z$ gate with $d-1$ control qubits.

Let $D_d(n)$ denote the subset of diagonal unitaries in $C_d(n)$. In fact, $D_d(n)$ are groups for all $d$ and $n$. Moreover, $D_d(n)$ can be generated by $Z^{1/2^{d-1}}$ and $C^{(i)}Z^{1/2^j}$ with $i+j=d-1$ \cite{zeng2008semi}. In \cite{cui2017diagonal}, the authors characterize unitaries in the diagonal Clifford hierarchies for qudits with prime power dimension. We reproduce one of these results for qubits.
\begin{theorem}\label{thm:diagonal_cliffords_FB}
For $d>1$, $V\in D_d(n)$ if and only if, up to a global phase, $V$ takes the form
\begin{equation}\label{eq:characterize_diagonal}
\exp\left(i\frac{\pi}{2^d}\sum_{\substack{S\subseteq[n]\\ S\neq\emptyset}}a_SZ^{S}\right)
\end{equation}
with $Z^{S}=\prod_{j\in S}Z_j$ and $a_S\in\mathbb{Z}$ for all $S\subseteq[n]$. We also have $V\not\in D_{d-1}(n)$ if and only if at least one $a_S$ is odd.
\end{theorem}
\begin{proof}
The reverse direction, that $V$ in the form of Eq.~\eqref{eq:characterize_diagonal} is in $D_d(n)$, is easy to show inductively.

The forward direction is a simple proof via contradiction. Suppose $V\in D_d(n)$. Any diagonal unitary can be written in the form of Eq.~\eqref{eq:characterize_diagonal} if we allow the $a_S$ to be real numbers.\footnote{Suppose $V=\sum_{T}e^{-i\pi\phi_T/2^n}\ket{T}\bra{T}$. Then we can choose the phases in Eq.~\eqref{eq:characterize_diagonal} to be $a_S=\frac{1}{2^N}\sum_{T}(-1)^{S\cdot T}\phi_T$ (treating $S,T\in\{0,1\}
^n$ as bit strings). The inverse is of course $\phi_T=\sum_S(-1)^{S\cdot T}a_S$.} So assume that for some $S_0\subseteq[n]$, $S_0\neq\emptyset$ we have $a_{S_0}\not\in\mathbb{Z}$. Let $i\in S_0$ and define a Clifford unitary $C$ to be a circuit of CX gates that maps $Z^{S_0}$ to $Z^{(i)}=Z_i$ and any other $Z^{S}$, $S\neq S_0$, to some $Z^{S'}$ with $i\not\in S'$.

Now, define $K_0=CVC^\dag$, $K_j=[K_{j-1},X_i]$. Since $V\in C_d(n)$, also $CVC^\dag\in C_d(n)$, and we must have $K_d=\pm I$. Calculating $K_d$ however we have
\begin{align}
K_1&=[CVC^\dag,X_i]=\exp\left(-i\frac{\pi}{2^{d-1}}a_{S_0}Z_i\right),\\
K_j&=K_{j-1}^2,\text{\space\space} j\ge2,\\
K_d&=\exp\left(-i\pi a_{S_0}Z_i\right).
\end{align}
We thus realize that $K_d$ can only be proportional to identity if $a_{S_0}$ is an integer, from which we get our contradiction.
\end{proof}

We then have the following statement regarding the states produced by circuits $V \in D_d(n)$.

\begin{prop}\label{prop:gps_diagonal_unitaries}
Let $V$ be an $n$-qubit quantum circuit belonging to $D_d(n)$, the subgroup of diagonal unitaries in the $d$-th level of the Clifford hierarchy $C_d(n)$. The state produced by the action of $V$ on $\ket{+}^{\otimes n}=\mathsf{H}^{\otimes n}\ket{0}^{\otimes n}$ (up to a global phase) is
\begin{equation}
    \ket{\psi_f} = V\ket{+}^{\otimes n} = \frac{1}{\sqrt{2^n}} \sum \limits_{x \in \Fset{n}} \omega_q^{f(x)} \ket{x},
\end{equation}
for some $f \in \calP_q(n,d)$ with $q=2^d$, with the $\FF_2$-representation
\begin{equation}
    f(x) = \sum_{\substack{T\subseteq[n]\\ 1\le|T|\le d}}c_T\prod_{j\in T}x_j \pmod{2^d}, \quad c_T \in 2^{|T|-1} \ZZ_{2^{d+1-|T|}}.
    \label{eq:f2_representation_diagonal_unitaries}
\end{equation}
\end{prop}
\begin{proof}
Applying a unitary from $D_d(n)$ to the state $\ket{+}^{\otimes n}$, we obtain a generalized phase state
\begin{align}\label{eq:Vappliedto+}
V\ket{+}^{\otimes n}&=\frac{1}{\sqrt{2^n}}\sum_{x\in\Fset{n}}\omega_{2^{d+1}}^{g(x)}\ket{x},\\
g(x)&=\sum_{\substack{S\subseteq[n]\\ S\neq\emptyset}}a_S\prod_{i\in S}(-1)^{x_i},\quad a_S\in\mathbb{Z}.
\end{align}
We can also understand this phase state by converting $g(x)$ to its $\FF_2$-representation using $(-1)^{x_j}=1-2x_j$ (as $x \in \Fset{n})$. Since $g(x)$ can be evaluated modulo $2^{d+1}$, monomials with degree greater than $d$ can be removed. We find
\begin{equation}
g(x) = \sum_{\substack{T\subseteq[n]\\ |T|\le d}}b_T\prod_{j\in T}x_j \pmod{2^{d+1}}, \quad b_T=(-2)^{|T|}\sum_{\substack{S\supseteq T\\ S\neq\emptyset}}a_S\in 2^{|T|}\mathbb{Z}.
\end{equation}

We note, from the $b_{\emptyset}$ term, that $V$ introduces a phase of $\omega_{2^{d+1}}^{b_{\emptyset}}$ to the basis state $\ket{0}^{\otimes n}$. Removing this, we obtain $\tilde g(x)=g(x)-b_{\emptyset}$, which is divisible by 2, i.e.~$f(x)=\tilde g(x)/2$ is a polynomial, because all $b_T$ for $T\neq\emptyset$ are even. Therefore, Eq.~\eqref{eq:Vappliedto+} becomes
\begin{equation}
\omega_{2^{d+1}}^{-b_\emptyset}V\ket{+}^{\otimes n}=\frac{1}{\sqrt{2^n}}\sum_{x\in\{0,1\}^n}\omega_{2^d}^{f(x)}\ket{x},
\end{equation}
which are exactly the states from Eq.~\eqref{eq:generalized_phase_state}, since $f(x)$ is the degree-$d$ polynomial
\begin{equation}
    f(x) = \sum_{\substack{T\subseteq[n]\\ 1\le|T|\le d}}c_T\prod_{j\in T}x_j \pmod{2^d}, \quad c_T=b_T/2\in 2^{|T|-1}\mathbb{Z}.
\end{equation}
\end{proof}
To learn a circuit representation of $V$ from samples, we have the following result.
\begin{theorem}
Let $V$ be an unknown $n$-qubit quantum circuit in $D_d(n)$, the group of diagonal unitaries in the $d$-th level of the Clifford hierarchy. A circuit representation of $V$ can then be learned through $O(n^d)$ queries to $V$ and using only separable measurements. 
\end{theorem}
\begin{proof}
From Proposition~\ref{prop:gps_diagonal_unitaries}, we note the correspondence between $V$ and the generalized phase state $\ket{\psi_f}$. Using Theorem~\ref{thm:generalizedphasedegreed}, we can learn the multi-linear representation of $f$~(Eq.~\eqref{eq:f2_representation_diagonal_unitaries}) corresponding to such a state, using separable measurements on $O(n^d)$ copies of $V\mathsf{H}^{\otimes n}\ket{0}^{\otimes n}$. We obtain a circuit representation of $V$ by inserting appropriate gates corresponding to monomials $\prod_{i \in T \subseteq [n]} x_i$, characterized by set $T$, present in $f$. We now define these gates with respect to the state $\ket{x}$ where $x \in \Fset{n}$. For $|T|\geq 2$, we insert a controlled-diagonal gate over qubits in $T$, that puts a phase of $\exp(i \pi c_T/2^{d-1})$ if $x_j=1,\, \forall j \in T$ and no phase otherwise. For monomials corresponding to singletons $x_j$ for $j\in [n]$ (i.e., $|T|=1$), we insert the phase gate $Z^{c_T/2^{d-1}} = \ket{0}\bra{0} + \exp(i \pi c_T/2^{d-1})\ket{1}\bra{1}$ on qubit $j$.
\end{proof}
\bibliographystyle{alpha}
\bibliography{references}

\end{document}